\documentclass[a4paper,UKenglish,cleveref, autoref, thm-restate,authorcolumns]{lipics-v2019}
\title{The Strahler Number of a Parity Game} %TODO Please add

\titlerunning{The Strahler Number of a Parity Game} %TODO optional, please use if title is longer than one line

\author{Laure Daviaud}{Department of Computer Science, City, University of London, UK}{Laure.Daviaud@city.ac.uk}{https://orcid.org/0000-0002-9220-7118}{}%(Optional) author-specific funding acknowledgements}%TODO mandatory, please use full name; only 1 author per \author macro; first two parameters are mandatory, other parameters can be empty. Please provide at least the name of the affiliation and the country. The full address is optional

\author{Marcin Jurdzi\'nski}{Department of Computer Science, University of Warwick, UK}{Marcin.Jurdzinski@warwick.ac.uk}{https://orcid.org/0000-0003-3640-8481}{}

\author{K.\,S. Thejaswini}{Department of Computer Science, University of Warwick, UK}{Thejaswini.Raghavan.1@warwick.ac.uk}{}{}

\authorrunning{L.~Daviaud, M.~Jurdzi\'nski, and K.\,S.~Thejaswini} %TODO mandatory. First: Use abbreviated first/middle names. Second (only in severe cases): Use first author plus 'et al.'


 %TODO mandatory; please add comma-separated list of keywords

\acknowledgements{This work has been supported by the EPSRC grant
  EP/P020992/1 (Solving Parity Games in Theory and
  Practice).
  It has started when the third author was supported by the Department
  of Computer Science at the University of Warwick as a visiting
  student from Chennai Mathematical Institute, India. 
  We thank colleagues in the following human chain for
  contributing to our education on the Strahler number of a tree: 
  the first and the second authors have learnt it from the third
  author, who has learnt it from K~Narayan Kumar, who has learnt it
  from Dmitry Chistikov, who has learnt it from Radu Iosif and Stefan
  Kiefer at a workshop on Infinite-State Systems, which has been
  organized and hosted by Jo\"el Ouaknine and Prakash Panangaden at
  Bellairs Research Institute.  
}%optional    

\nolinenumbers %uncomment to disable line numbering

\hideLIPIcs  %uncomment to remove references to LIPIcs series (logo, DOI, ...), e.g. when preparing a pre-final version to be uploaded to arXiv or another public repository

\newcommand{\Bc}{\mathcal{B}}
\newcommand{\Cc}{\mathcal{C}}
\newcommand{\Dc}{\mathcal{D}}
\newcommand{\Gc}{\mathcal{G}}
\newcommand{\Hc}{\mathcal{H}}
\newcommand{\Lc}{\mathcal{L}}

\newcommand{\Rc}{\mathcal{R}}
\newcommand{\Tc}{\mathcal{T}}
\newcommand{\Uc}{\mathcal{U}}
\newcommand{\Vc}{\mathcal{V}}

\newcommand{\Even}{\mathrm{Even}}
\newcommand{\Odd}{\mathrm{Odd}}
\newcommand{\Reg}[2]{\Rc^{#1}\!\left(#2\right)}
\newcommand{\Def}[2]{\Dc^{#1}\!\left(#2\right)}
\newcommand{\Strah}[1]{\mathrm{Str}\left(#1\right)}
\newcommand{\height}[1]{\mathrm{height}\left(#1\right)}
\newcommand{\leaves}[1]{\mathrm{leaves}\left(#1\right)}
\newcommand{\floor}[1]{\left\lfloor #1 \right\rfloor}
\newcommand{\ceil}[1]{\left\lceil #1 \right\rceil}
\newcommand{\seq}[1]{\left\langle #1 \right\rangle}
\newcommand{\tpl}[1]{\left( #1 \right)}
\newcommand{\eset}[1]{\left\{\, #1 \,\right\}}
\newcommand{\Nats}{\mathbb{N}}

\theoremstyle{remark}
\newtheorem{question}[theorem]{Question}
\theoremstyle{plain}

\begin{document}

\maketitle

\begin{abstract}
  The Strahler number of a rooted tree is the largest height of a
  perfect binary tree that is its minor. 
  The Strahler number of a parity game is proposed to be defined as
  the smallest Strahler number of the tree of any of its attractor
  decompositions.  
  It is proved that parity games can be solved in quasi-linear space
  and in time that is polynomial in the number of vertices~$n$ and
  linear in~$({d}/{2k})^k$, where $d$ is the number of priorities and
  $k$ is the Strahler number. 
  This complexity is quasi-polynomial because the Strahler number is
  at most logarithmic in the number of vertices. 
  The proof is based on a new construction of small Strahler-universal
  trees.  

  It is shown that the Strahler number of a parity game is a robust, 
  and hence arguably natural, parameter:  
  it coincides with its alternative version based on trees of progress
  measures and---remarkably---with the register number 
  defined by Lehtinen~(2018).  
  It follows that parity games can be solved in quasi-linear space and
  in time that is polynomial in the number of vertices and linear 
  in~$({d}/{2k})^k$, where $k$ is the register number.
  This significantly improves the running times and space achieved for
  parity games of bounded register number by Lehtinen~(2018) and by
  Parys~(2020). 

  The running time of the algorithm based on small Strahler-universal
  trees yields a novel trade-off $k \cdot \lg(d/k) = O(\log n)$ 
  between the two natural parameters that measure the structural
  complexity of a parity game, 
  which allows solving parity games in polynomial time. 
  This includes as special cases the asymptotic settings of those
  parameters covered by the results of Calude, Jain Khoussainov, Li,
  and Stephan~(2017), of Jurdzi\'nski and Lazi\'c~(2017), and of 
  Lehtinen~(2018), and it significantly extends the range of such
  settings, for example to $d = 2^{O\left(\sqrt{\lg n}\right)}$ and 
  $k = O\!\left(\sqrt{\lg n}\right)$. 
\end{abstract}

\newpage

\section{Context}

\subparagraph*{Parity Games.}
Parity games are a fundamental model in automata theory and
logic~\cite{EJ91,Zie98,GTW01,BW18}, and their applications to
verification, program analysis, and synthesis. 
In particular, they are intimately linked to the problems of emptiness
and complementation of non-deterministic automata on
trees~\cite{EJ91,Zie98}, model checking and satisfiability of fixpoint
logics~\cite{EJS93,BW18}, and evaluation of nested fixpoint
expressions~\cite{BKMP19,HS19}. 
It is a long-standing open problem whether parity games can be solved
in polynomial time~\cite{EJS93}. 

The impact of parity games goes well beyond their home turf of
automata theory, logic, and formal methods.
For example, an answer~\cite{Fri09} of a question posed originally for
parity games~\cite{VJ00} has strongly inspired major breakthroughs on
the computational complexity of fundamental algorithms in stochastic 
planning~\cite{Fea10} and linear optimization~\cite{Fri11,FHZ11}.

\subparagraph*{Strahler Number.}
The Strahler number has been proposed by Horton (1945) and made 
rigorous by Strahler~(1952), in their morphological study of river
networks in hydrogeology. 
It has been also studied in other sciences, such as botany, anatomy, 
neurophysiology, physics, and molecular biology, where branching
patterns appear.  
The Strahler number has been identified in computer science by
Ershov~\cite{Ers58} as the smallest number of registers needed to
evaluate an arithmetic expression. 
It has since been rediscovered many times in various areas of computer
science;
see the surveys of Knuth~\cite{Knu73}, Viennot~\cite{Vie90}, and 
Esparza, Luttenberger, and Schlund~\cite{ELS16}.

\subparagraph{Related Work.}
A major breakthrough in the quest for a polynomial-time algorithm for
parity games was achieved by Calude, Jain, Khoussainov, Li, and
Stephan~\cite{CJKLS17}, who have given the first quasi-polynomial
algorithm. 
Other quasi-polynomial algorithm have been developed soon after by
Jurdzi\'nski and Lazi\'c~\cite{JL17}, and Lehtinen~\cite{Leh18}.  
Czerwi\'nski, Daviaud, Fijalkow, Jurdzi\'nski, Lazi\'c, and
Parys~\cite{CDFJLP19} have introduced the concepts of 
\emph{universal trees} and \emph{separating automata}, and  argued
that all the aforementioned quasi-polynomial algorithms were intimately 
linked to them. 

By establishing a quasi-polynomial lower bound on the size of
universal trees, Czerwi\'nski et al.\ have highlighted the fundamental
limitations of the above approaches, motivating further the study of
the attractor decomposition algorithm due to McNaughton~\cite{McN93}
and Zielonka~\cite{Zie98}. 
Parys~\cite{Par19} has proposed an ingenious quasi-polynomial version
of McNaughton-Zielonka algorithm, but Lehtinen, Schewe, and
Wojtczak~\cite{LSW19}, and Jurdzi\'nski and Morvan~\cite{JM20} have
again strongly linked all quasi-polynomial variants of the
attractor decomposition algorithm to universal trees.  

Among several prominent quasi-polynomial algorithms for parity games,
Lehtinen's approach~\cite{Leh18} has relatively least attractive
worst-case running time bounds.
Parys~\cite{Par20} has offered some running-time improvements to
Lehtinen's algorithm, but it remains significantly worse than
state-of-the-art bounds of Jurdzi\'nski and Lazi\'c~\cite{JL17}, and
Fearnley, Jain, de Keijzer, Schewe, Stephan, and
Wojtczak~\cite{FJKSSW19}, in particular because it always requires at 
least quasi-polynomial working space.

\subparagraph{Our Contributions.}
We propose the Strahler number as a parameter that measures the
structural complexity of dominia in a parity game and that governs the 
computational complexity of the most efficient algorithms currently
known for solving parity games.  
We establish that the Strahler number is a robust, and hence natural, 
parameter by proving that it coincides with its version based on trees
of progress measures and with the register number defined by
Lehtinen~\cite{Leh18}.  

We give a construction of small Strahler-universal trees that, when
used with the progress measure lifting algorithm~\cite{Jur00,JL17} or
with the universal attractor decomposition algorithm~\cite{JM20},
yield algorithms that work in quasi-linear space and
quasi-polynomial time.
Moreover, usage of our small Strahler-universal trees allows to solve
parity games in polynomial time for a wider range of asymptotic
settings of the two natural structural complexity parameters  
(number of priorities~$d$ and the Strahler/register number~$k$) 
than previously known, and that covers as special cases the $k = O(1)$ 
criterion of Lehtinen~\cite{Leh18} and the $d < \lg n$  and $d =
O(\log n)$  criteria of of Calude et al.~\cite{CJKLS17}, and of
Jurdzi\'nski and Lazi\'c~\cite{JL17}, respectively.

\section{Dominions, Attractor Decompositions, and Their Trees}
\label{section:tuning}

\subparagraph*{Strategies, Traps, and Dominions.}

A \emph{parity game}~\cite{EJ91} $\Gc$ consists of a finite directed
graph~$(V, E)$, a partition $(V_{\Even}, V_{\Odd})$ of the set of
vertices~$V$, and a function $\pi : V \to \eset{0, 1, \dots, d}$ that
labels every vertex~$v \in V$ with a non-negative integer~$\pi(v)$
called its \emph{priority}.  
We say that a cycle is \emph{even} if the highest vertex priority on
the cycle is even; otherwise the cycle is \emph{odd}.
We say that a parity game is \emph{$(n, d)$-small} if it has at
most~$n$ vertices and all vertex priorities are at most~$d$. 

For a set~$S$ of vertices, we write $\Gc \cap S$ 
for the substructure of~$\Gc$ whose graph is the subgraph of~$(V, E)$
induced by the sets of vertices~$S$.
Sometimes, we also write $\Gc \setminus S$ to denote 
$\Gc \cap (V \setminus S)$. 
We assume throughout that every vertex has at least one outgoing edge, 
and we reserve the term \emph{subgame} to substructures $\Gc \cap S$, 
such that every vertex in the subgraph of $(V, E)$ induced by~$S$ has
at least one outgoing edge.

A (positional) \emph{Steven strategy} is a set $\sigma \subseteq E$ 
of edges such that:
\begin{itemize}
\item
  for every $v \in V_{\Even}$, there is an edge $(v, u) \in \sigma$,
\item
  for every $v \in V_{\Odd}$, if $(v, u) \in E$ then $(v, u) \in \sigma$.
\end{itemize}
For a non-empty set of vertices $R$, we say that a Steven
strategy~$\sigma$ \emph{traps Audrey in $R$} if  
$w \in R$ and $(w, u) \in \sigma$ imply $u \in R$. 
We say that a set of vertices~$R$ is a 
\emph{trap for Audrey}~\cite{Zie98} if there is a Steven strategy that
traps Audrey in~$R$.  
Observe that if~$R$ is a trap in a game~$\Gc$ then $\Gc \cap R$ is a
subgame of~$\Gc$.
For a set of vertices $D \subseteq V$, we say that a Steven
strategy~$\sigma$ is a 
\emph{Steven dominion strategy on $D$} if
  $\sigma$ traps Audrey in~$D$ and 
  every cycle in the subgraph $(D, \sigma)$ is even.
Finally, we say that a set~$D$ of vertices is a
\emph{Steven dominion}~\cite{JPZ08} if there is a Steven dominion
strategy on it. 

Audrey strategies, trapping Steven, and Audrey dominions are defined in
an analogous way by swapping the roles of the two players. 
We note that the sets of Steven dominions and of Audrey dominions are each  
closed under union, and hence the largest Steven and Audrey dominions
exist, and they are the unions of all Steven and Audrey dominions,
respectively. 
Moreover, every Steven dominion is disjoint from every Audrey
dominion.

\subparagraph*{Attractor Decompositions.}

In a parity game~$\Gc$, for a target set of vertices~$B$
(``bullseye'') and a set of vertices~$A$ such that $B \subseteq A$, 
we say that a Steven strategy~$\sigma$ is a 
\emph{Steven reachability strategy to $B$ from~$A$} if every infinite  
path in the subgraph $(V, \sigma)$ that starts from a vertex in~$A$
contains at least one vertex in~$B$.

For every target set~$B$, there is the largest
(with respect to set inclusion) set from which there is a Steven
reachability strategy to~$B$ in~$\Gc$;
we call this set the 
\emph{Steven attractor to~$B$ in~$\Gc$}~\cite{Zie98}. 
\emph{Audrey reachability strategies} and \emph{Audrey attractors} are 
defined analogously.
We highlight the simple fact that if~$A$ is an attractor for a player
in~$\Gc$ then its complement $V \setminus A$ is a trap for them.

If $\Gc$ is a parity game in which all priorities do not exceed a
non-negative even number~$d$ then we say that 
$\Hc \: = \:
\seq{A, (S_1, \Hc_1, A_1), \dots, (S_\ell, \Hc_\ell, A_\ell)}$
is a \emph{Steven $d$-attractor decomposition}~\cite{DJL18,DJL19,JM20}
of~$\Gc$ if: 
\begin{itemize}
\item
  $A$ is the Steven attractor to the (possibly empty) set of vertices
  of priority~$d$ in~$\Gc$;
\end{itemize}
and setting $\Gc_1 = \Gc \setminus A$, for all $i = 1, 2, \dots, \ell$,
we have: 
\begin{itemize}
\item
  $S_i$ is a non-empty trap for Audrey in~$\Gc_i$ in which every
  vertex priority is at most~$d-2$;  
\item
  $\Hc_i$ is a Steven $(d-2)$-attractor decomposition of 
  subgame~$\Gc \cap S_i$; 
\item
  $A_i$ is the Steven attractor to $S_i$ in~$\Gc_i$;
\item
  $\Gc_{i+1} = \Gc_i \setminus A_i$;
\end{itemize}
and the game $\Gc_{\ell+1}$ is empty.
If $d = 0$ then we require that $\ell = 0$.

The following proposition states that if a subgame induced by a trap
for Audrey has a Steven attractor decomposition then the trap is a
Steven dominion. 
Indeed, a routine proof argues that the union of all the Steven
reachability strategies, implicit in the attractors listed in the
decomposition, is a Steven dominion strategy. 

\begin{proposition}[\cite{Zie98,DJL18,JM20}]
\label{prop:decomposition-dominion-even}
  If $d$ is even, $R$ is a trap for Audrey in~$\Gc$, and there is
  a Steven $d$-attractor decomposition of~$\Gc \cap R$, then $R$ is a
  Steven dominion in~$\Gc$. 
\end{proposition}
Attractor decompositions for Audrey can be defined in the analogous
way by swapping the roles of players as expected, and then a dual
version of the proposition holds routinely.

The following theorem implies that every vertex in a parity game is
either in the largest Steven dominion or in the largest Audrey
dominion---it is often referred to as the 
\emph{positional determinacy theorem} for parity games. 

\begin{theorem}[\cite{EJ91,McN93,Zie98,JM20}]
\label{thm:attractor-decompositions-of-largest-dominia}
  For every parity game~$\Gc$, there is a partition of the set of
  vertices into a trap for Audrey~$W_{\Even}$ and a trap for
  Steven~$W_{\Odd}$, such that there is a Steven attractor
  decomposition of $\Gc \cap W_{\Even}$ and an Audrey attractor
  decomposition of $\Gc \cap W_{\Odd}$.   
\end{theorem}

\subparagraph*{Ordered Trees and Their Strahler Numbers.}

Ordered trees are defined inductively;
the trivial tree $\seq{}$ is an ordered tree and so is a sequence
$\seq{T_1, T_2, \dots, T_\ell}$, where $T_i$ is an ordered 
tree for every $i = 1, 2, \dots, \ell$.
The trivial tree has only one node called the root, which is a leaf;
and a tree of the form $\seq{T_1, T_2, \dots, T_\ell}$ has the root
with $k$ children, the root is not a leaf, and the $i$-th child of the
root is the root of ordered tree~$T_i$.  

Because the trivial tree~$\seq{}$ has just one node, we sometimes
write~$\circ$ to denote it. 
If $T$ is an ordered tree and $i$ is a positive integer, then 
we use the notation $T^i$ to denote the sequence 
$T, T, \dots, T$ consisting of $i$ copies of tree~$T$. 
Then the expression $\seq{T^i} = \seq{T, \dots, T}$ denotes 
the tree whose root has~$i$ children, each of which is the root of a
copy of~$T$.  
We also use the $\cdot$ symbol to denote concatenation of sequences,
which in the context of ordered trees can be interpreted as sequential
composition of trees by merging their roots;
for example, 
$\seq{\seq{\circ^3}} \cdot \seq{\circ^4, \seq{\seq{\circ}}^2} 
= 
\seq{\seq{\circ^3}, \circ^4, \seq{\seq{\circ}}^2}
= 
\seq{\seq{\circ, \circ, \circ}, \circ, \circ, \circ, \circ, \seq{\seq{\circ}}, \seq{\seq{\circ}}}$.

For an ordered tree~$T$, we write $\height{T}$ for its
\emph{height} and $\leaves{T}$ for its 
\emph{number of leaves}, which are defined by the following routine 
induction: 
the trivial tree~$\seq{} = \circ$ has $1$ leaf and its height is~$1$; 
the number of leaves of tree $\seq{T_1, T_2, \dots, T_\ell}$ is
the sum of the numbers of leaves of trees~$T_1$, $T_2$, \dots, 
$T_\ell$; 
and its height 
is $1$ plus the maximum height of trees~$T_1$, $T_2$, \dots, $T_\ell$. 
For example, the tree 
$\seq{\seq{\circ^3}, \circ^4, \seq{\seq{\circ}}^2}$
has $9$ leaves and height~$4$
We say that an ordered tree is \emph{$(n, h)$-small}
if it has at most $n$ leaves and its height is at most~$h$.

The \emph{Strahler number} $\Strah{T}$ of a tree~$T$ is defined to
be the largest height of a perfect binary tree that is a minor
of~$T$.  
Alternatively, it can be defined by the following structural
induction: 
the Strahler number of the trivial tree~$\seq{} = \circ$ is~$1$; 
and if $T = \seq{T_1, \dots, T_\ell}$ and 
$m$ is the largest Strahler number of trees~$T_1, \dots, T_\ell$, 
then $\Strah{T} = m$ if there is a unique $i$ such that 
$\Strah{T_i} = m$, 
and $\Strah{T} = m+1$ otherwise. 
For example, 
  we have 
  $\Strah{\seq{\seq{\circ^3}, \circ^4, \seq{\seq{\circ}}^2}} = 2$
  because $\Strah{\circ} = \Strah{\seq{\seq{\circ}}} = 1$ and
  $\Strah{\seq{\circ^3}} = 2$.

\begin{proposition}
\label{prop:Strahler-small}
  For every $(n, h)$-small tree~$T$, we have 
  $\Strah{T} \leq h$ and $\Strah{T} \leq \lfloor \lg n \rfloor + 1$. 
\end{proposition}

\subparagraph*{Trees of Attractor Decompositions.}

The definition of an attractor decomposition is inductive and we
define an ordered tree that reflects the hierarchical structure of an
attractor decomposition.
If $d$ is even and 
$\Hc = \seq{A, (S_1, \Hc_1, A_1), \dots, (S_\ell, \Hc_\ell, A_\ell)}$ 
is a Steven $d$-attractor decomposition then we define the 
\emph{tree of attractor decomposition~$\Hc$}~\cite{DJL19,JM20}, 
denoted by $T_{\Hc}$, to be the trivial ordered tree~$\seq{}$ if 
$\ell = 0$, and otherwise, to be the ordered tree  
$\seq{T_{\Hc_1}, T_{\Hc_2}, \dots, T_{\Hc_\ell}}$, where for every
$i = 1, 2, \dots, \ell$, tree $T_{\Hc_i}$ is the tree of attractor 
decomposition~$\Hc_i$.  
Trees of Audrey attractor decompositions are defined analogously.

Observe that the sets $S_1, S_2, \dots, S_\ell$ in an attractor
decomposition as above are non-empty and pairwise disjoint, which
implies that trees of attractor decompositions are small relative to
the number of vertices and the number of distinct priorities in a
parity game. 
The following proposition can be proved by routine structural
induction.  

\begin{proposition}[\cite{DJL19,JM20}]
\label{prop:tree-of-decomposition-is-small}
  If $\Hc$ is an attractor decomposition of an $(n, d)$-small parity
  game 
  then its tree $T_{\Hc}$ is $(n, \ceil{d/2}+1)$-small.
\end{proposition}

We define the 
\emph{Strahler number of an attractor decomposition~$\Hc$}, denoted 
by~$\Strah{\Hc}$, to be the Strahler number 
$\Strah{T_{\Hc}}$ of its tree~$T_{\Hc}$.   
We define the \emph{Strahler number of a parity game} to be the
maximum of the smallest Strahler numbers of attractor decompositions
of the largest Steven and Audrey dominions, respectively.

\section{Strahler Strategies in Register Games}

This section establishes a connection between the register number of a
parity game defined by Lehtinen~\cite{Leh18} and the Strahler number. 
More specifically, we argue that from every Steven attractor
decomposition of Strahler number~$k$, we can derive a dominion
strategy for Steven in the $k$-register game.
Once we establish the Strahler number upper bound on the register
number, we are faced with the following two natural questions: 

\begin{question}
\label{question:Strahler-eq-register}
  Do the Strahler and the register numbers coincide?
\end{question}
\begin{question}
\label{question:Strahler-algorithmic}
  Can the relationship between Strahler and register numbers be
  exploited algorithmically, in particular, to improve the running
  time and space complexity of solving register games studied by 
  Lehtinen~\cite{Leh18} and Parys~\cite{Par20}?
\end{question}
This work has been motivated by those two questions and it answers 
them both positively 
(Lemma~\ref{lem:Strahler-bounds-Lehtinen} and
Theorem~\ref{thm:Lehtinen-bounds-Strahler}, 
and Theorem~\ref{thm:Strahler-pm-run-time}, respectively).

For every positive number~$k$, a Steven \emph{$k$-register game}
on a parity game~$\Gc$ is another parity game $\Reg{k}{\Gc}$ whose 
vertices, edges, and priorities will be referred to as \emph{states},
\emph{moves}, and \emph{ranks}, respectively, for disambiguation.   
The states of the Steven $k$-register game on~$\Gc$ are either
pairs $\left(v, \seq{r_{k}, r_{k-1}, \dots, r_1}\right)$ or
triples $\left(v, \seq{r_{k}, r_{k-1}, \dots, r_1}, p\right)$,
where $v$ is a vertex in~$\Gc$, 
$d \geq r_{k} \geq r_{k-1} \geq \cdots \geq r_1 \geq 0$, and  
$1 \leq p \leq 2k+1$. 
The former states have rank~$1$ and the latter have rank~$p$.  
Each number $r_i$, for $i = k, k-1, \dots, 1$, is referred to as 
the value of the $i$-th register in the state. 
Steven owns all states 
$\left(v, \seq{r_{k}, r_{k-1}, \dots, r_1}\right)$ and the owner
of vertex~$v$ in~$\Gc$ is the owner of states
$\left(v, \seq{r_{k}, r_{k-1}, \dots, r_1}, p\right)$ for
every~$p$.  
How the game is played by Steven and Audrey is determined by the
available moves:
\begin{itemize}
\item
  at every state 
  $\left(v, \seq{r_{k}, r_{k-1}, \dots, r_1}\right)$, 
  Steven picks $i$, such that $0 \leq i \leq k$, and \emph{resets} 
  registers $i, i-1, i-2, \dots, 1$, leading to state 
  $\left(v, 
  \seq{r'_{k}, \dots, r'_{i+1}, r'_i, 0, \dots, 0}, p\right)$ 
  of rank~$p$ and with updated register values, where:
  \[
  p \: = \:
  \begin{cases}
    2i & 
    \text{if $i \geq 1$ and $\max\left(r_i, \pi(v)\right)$ is even}, 
    \\
    2i+1 & 
    \text{if $i = 0$, or if $i \geq 1$ and 
      $\max\left(r_i, \pi(v)\right)$ is odd}; 
  \end{cases}
  \]
  $r'_j = \max\!\left(r_j, \pi(v)\right)$ for $j \geq i+1$, 
  and $r'_i = \pi(v)$; 

\item
  at every state 
  $\left(v, \seq{r_{k}, r_{k-1}, \dots, r_1}, p\right)$, 
  the owner of vertex~$v$ in~$\Gc$ picks an edge $(v, u)$ in~$\Gc$,
  leading to 
  state~$\left(u, \seq{r_{k}, r_{k-1}, \dots, r_1}\right)$ of
  rank~$1$ and with unchanged register values. 
\end{itemize}
For example, at state 
$\left(v, \seq{9, 6, 4, 4, 3}\right)$ of rank~$1$, if the priority
$\pi(v)$ of vertex~$v$ is~$5$ and 
Steven picks~$i = 3$, this leads to state 
$\left(v, \seq{9, 6, 5, 0, 0}, 7\right)$ of rank~$2i+1 = 7$ because 
$\max\!\left(r_3, \pi(v)\right) = \max(4, 5) = 5$ is odd, 
$r'_4 = \max(r_4, \pi(v)) = \max(6, 5) = 6$, and 
$r'_3 = \pi(v) = 5$. 

Observe that the first components of states on every cycle in
game~$\Reg{k}{\Gc}$ form a (not necessarily simple) cycle in parity
game~$\Gc$;
we call it the cycle in~$\Gc$ \emph{induced} by the cycle
in~$\Reg{k}{\Gc}$.   
If a cycle in~$\Reg{k}{\Gc}$ is even
(that is, the highest state rank on it is even)
then the induced cycle in~$\Gc$ is also even.
Lehtinen~\cite[Lemmas~3.3 and~3.4]{Leh18} has shown that a vertex~$v$
is in the largest Steven dominion in~$\Gc$ if and only if there is a
positive integer~$k$ such that a state $\tpl{v, \overline{r}}$, for
some register values~$\overline{r}$ is in the largest Steven dominion
in~$\Reg{k}{\Gc}$.  
Lehtinen and Boker~\cite[a comment after Definition~3.1]{LB20} have
further clarified that for every~$k$, if a player has a dominion
strategy in~$\Reg{k}{\Gc}$ from a state whose first component is a
vertex~$v$ in~$\Gc$, then they also have a dominion strategy
in~$\Reg{k}{\Gc}$ from every state whose first component is~$v$.
This allows us to say without loss of rigour that a vertex~$v$
in~$\Gc$ is in a dominion in~$\Reg{k}{\Gc}$. 

By defining the 
\emph{(Steven) register number}~\cite[Definition~3.5]{Leh18} 
of a parity game~$\Gc$ to be the smallest number~$k$ such that  
all vertices $v$ in the largest Steven dominion in~$\Gc$ are in a
Steven dominion in~$\Reg{k}{\Gc}$, 
and by proving the $1 + \lg n$ upper bound on the register number of
every $(n, d)$-small parity game~\cite[Theorem~4.7]{Leh18}, Lehtinen
has contributed a novel quasi-polynomial algorithm for solving parity
games, adding to those by Calude et al.~\cite{CJKLS17} and
Jurdzi\'nski and Lazi\'c~\cite{JL17}.

Lehtinen~\cite[Definition~4.8]{Leh18} has also considered the concept
of a Steven \emph{defensive dominion strategy} in a
$k$-register game (for brevity, we call it a $k$-defensive strategy): 
it is a Steven dominion strategy on a set of states
in~$\Reg{k}{\Gc}$ in which there is no state of rank~$2k+1$. 
Alternatively, the same concept can be formalized by defining the
\emph{defensive $k$-register game} $\Def{k}{\Gc}$, which is played
exactly like the $k$-register game~$\Reg{k}{\Gc}$, but in which Audrey
can also win just by reaching a state of rank~$2k+1$.
Note that the game $\Def{k}{\Gc}$ can be thought of as having the
winning criterion for Steven as being a conjunction of a parity
and a safety criteria, and the winning criterion for Audrey
as a disjunction of a parity and a reachability criteria.
Routine arguements allow to extend positional determinacy from parity
games to such games with combinations of parity, and safety or
reachability winning criteria.

We follow Lehtinen~\cite[Definition~4.9]{Leh18} by defining the 
\emph{(Steven) defensive register number} of a Steven dominion~$D$
in~$\Gc$ as the smallest number~$k$ such that Steven has a defensive
dominion strategy in~$\Reg{k}{\Gc}$ on a set of states that 
includes all $\tpl{v, \seq{r_{k}, \dots, r_1}}$ for $v \in D$, 
and such that  $r_{k}$ is an even number at least as large as every
vertex priority in~$D$. 
We propose to call it the \emph{Lehtinen number} of a Steven dominion
in~$\Gc$ to honour Lehtinen's insight that led to this---as we argue
in this work---fundamental concept. 
We also define the Lehtinen number of a vertex in~$\Gc$ to be the
smallest Lehtinen number of a Steven dominion in~$\Gc$ that includes
the vertex, and the Lehtinen number of a parity game to be the
Lehtinen number of its largest Steven dominion. 
We also note that the register and the Lehtinen numbers of a parity
game nearly coincide (they differ by at most one), and hence the
conclusions of our analysis of the latter also apply to the former.

\begin{lemma}
\label{lem:Strahler-bounds-Lehtinen}
  The Lehtinen number of a parity game is 
  no larger than 
  its Strahler number. 
\end{lemma}
The arguments used in our proof of this lemma are similar to those
used in the proof of the main result of
Lehtinen~\cite[Theorem~4.7]{Leh18}.  
Our contribution here is to pinpoint the Strahler number of an
attractor decomposition as the structural parameter of a dominion that
naturally bounds the number of registers used in Lehtinen's
construction of a defensive dominion strategy.

\begin{proof}[Proof of Lemma~\ref{lem:Strahler-bounds-Lehtinen}]
  Consider a parity game $\Gc$ and let $d$ be the least even integer
  no smaller than any of the priority in~$\Gc$. Consider a Steven
  d-attractor decomposition $\Hc$ of $\Gc$ of Strahler number $k$. We
  construct a defensive $k$-register strategy for Steven
  on~$\Reg{k}{\Gc}$. The strategy is defined inductively on the height
  of $\Tc_\Hc$, and has the additional property of being
  \textit{$\Gc$-positional} in the following sense: if $\left(\left(v,
  \seq{r_{k}, \dots, r_1}\right), \left(v, \seq{r'_{k}, \dots, r'_1},p
  \right) \right)$ is a move then the register reset by Steven only
  depends on $v$, not on the values in the registers. Similarly, if
  $\left(\left(v, \seq{r_{k}, \dots, r_1}, p\right), \left(u,
  \seq{r_{k}, \dots, r_1} \right) \right)$ is a move and $v$ is owned
  by Steven, $u$ only depends on $v$ and not on the values of the
  registers or $p$.  

  \subparagraph{Strategy for Steven.}
  If $\Hc = \seq{A,\emptyset}$, then $\Gc$ consists of the set of
  vertices of priority $d$ and of its Steven attractor. In this case,
  Steven follows the strategy induced by the reachability strategy in
  $A$ to the set of vertices of priority $d$, only resetting register
  $r_1$ immediately after visiting a state with first component a
  vertex of priority $d$ in~$\Gc$.  
  More precisely, the Steven defensive strategy is defined with the
  following moves: 
  \begin{itemize}
  \item 
    $\left( \left(v, \seq{r_1}\right), \left(v, \seq{r_1}, 1\right)
    \right)$ if $v$ is not a vertex of priority $d$ in~$\Gc$; 
  \item 
    $\left( \left(v, \seq{r_1}\right), \left(v, \seq{r'_1}, 2\right)
    \right)$ if $v$ is a vertex of priority $d$ in~$\Gc$ and $r'_1 =
    \max(r_1,d)$ is even; 
  \item 
    $\left( \left(v, \seq{r_1}\right), \left(v, \seq{r'_1}, 3\right)
    \right)$ if $v$ is a vertex of priority $d$ in~$\Gc$ and $r'_1 =
    \max(r_1,d)$ is odd (we state this case for completeness but this
    will never occur); 
  \item 
    $\left(\left(v, \seq{r_1}, p\right), \left(u, \seq{r_1}\right)
    \right)$ where $(v,u)$ belongs to the Steven reachability strategy
    from $A$ to the set of vertices of priority $d$ in~$\Gc$. 
  \end{itemize} 
  Note that this strategy is $\Gc$-positional.

  Suppose now that 
  $\Hc \: = \: \seq{A, (S_1, \Hc_1, A_1), \dots,
    (S_\ell, \Hc_\ell, A_\ell)}$ 
  and that it has Strahler number~$k$. 
  For all $i = 1, 2, \dots, \ell$, let $k_i$ be the Strahler number
  of~$\Hc_i$.  
  By induction, for all $i$, we have a Steven defensive $k_i$-register
  strategy $\sigma_i$, which is $(\Gc \cap S_i)$-positional, on a set
  of states $\Omega_i$ in~$\Reg{k_i}{\Gc \cap S_i}$ including all the
  states $\left(v, \seq{r_{k_i}, \dots, r_1}\right)$ for $v \in S_i$ and
  $r_{k_i}$ an even number at least as large as every vertex  
  priority in~$S_i$. 
Let $\Gamma_i$ be the set of states in~$\Reg{k}{\Gc \cap S_i}$ defined as all the states $\left(v, \seq{d, r_{k-1}, \dots, r_1}\right)$ for $v \in S_i$ if $k_i \neq k$ and as the union of the states $\left(v, \seq{d, r_{k-1}, \dots, r_1}\right)$ for $v \in S_i$ and $\Omega_i$, otherwise.

The strategy $\sigma_i$ induces a strategy
on $\Gamma_i$ in~$\Reg{k}{\Gc \cap S_i}$ by simply ignoring registers $r_{k_i+1}, \ldots, r_{k}$, and using $(\Gc \cap S_i)$-positionality to define moves from the states not in $\Omega_i$. More precisely, in a state $\left(v, \seq{r_{k}, \ldots, r_1}\right)$, Steven resets register $j$ if and only if register $j$ is reset in a state $\left(v, \seq{r'_{k_i}, \dots, r'_1}\right)$ of $\Omega_i$ according to $\sigma_i$. This is well defined by $(\Gc \cap S_i)$-positionality. Similarly, we add moves $\left(\left(v, \seq{r_{k}, \dots, r_1}, p\right), \left(u, \seq{r_{k}, \dots, r_1} \right) \right)$ to the strategy if and only if there is a move $\left(\left(v, \seq{r'_{k_i}, \dots, r'_1},p'\right), \left(u, \seq{r'_{k_i}, \dots, r'_1}\right) \right)$ in $\sigma_i$. This is again well-defined by $(\Gc \cap S_i)$-positionality.

This strategy is denoted by $\tau_i$. Note that $\tau_i$ is a defensive $k$-register strategy on $\Gamma_i$, which is $\Gc$-positional.

The Steven defensive strategy in~$\Reg{k}{\Gc}$ is defined by the following moves, where $S$ denotes the set of vertices of priority $d$ in $\Gc$:
  \begin{itemize}
  \item On the set of states with first component a vertex of $A_i\setminus S_i$, the moves are given by $\tau_i$.
  \item On the set of states with first component a vertex of $A\setminus S$, Steven uses the strategy induced by the reachability strategy from $A_i$ to $S_i$, without resetting any registers.
  \item On $\Reg{k}{\Gc \cap (A\setminus S)}$, Steven uses the strategy induced by the reachability strategy from $A$ to $S$, without resetting any registers.
  \item On the set of states with first component a vertex of $S$, 
  \begin{itemize}
  \item $\left( \left(v, \seq{r_k, \ldots, r_1}\right), \left(v, \seq{d, 0, \ldots, 0}, p\right) \right)$ where $v$ is a vertex in $S$ and $p=2k$ if $\max(r_k,d)$ is even and $p=2k+1$ otherwise. 
  \item $\left( \left(v, \seq{r_k, \ldots, r_1},p\right), \left(u, \seq{r_k, \ldots, r_1}\right) \right)$ for some uniquely chosen $u$ such that $(v,u)$ in $E$ if $v$ is owned by Steven and for all $u$ such that $(v,u)$ in $E$ if $v$ is owned by Audrey.
  \end{itemize}
  \end{itemize} 
Observe that this strategy is $\Gc$-positional.

\subparagraph{Correctness of the Strategy.} 
We prove now that the strategy defined above is indeed a defensive
$k$-register strategy. 
We proceed by induction on the height of $\Tc_\Hc$ and define a set of
states $\Gamma$, including all the states 
$\left(v, \seq{d,r_{k-1}, \dots, r_1}\right)$ such that $v$ is a
vertex of $\Gc$.    

  \textit{Base Case:}
  If the height of $\Tc_\Hc$ is $0$ and $\Hc = \seq{A,\emptyset}$, let $\Gamma$ be the set of states $\left(v, \seq{r_1}\right)$ and $\left(v, \seq{r_1} , p\right)$ with $v$ a vertex of $\Gc$, $1\leq r_1 \leq d$ and $p$ being either $1$ or $2$. It is easy to see that the strategy defined above is a defensive dominion strategy on this set.

  \textit{Inductive step:} If $\Hc \: = \: \seq{A, (S_1, \Hc_1, A_1), \dots, (S_\ell, \Hc_\ell, A_\ell)}$ with Strahler number $k$ and $k_i$ being the Strahler number of $\Hc_i$ for all $i$ (note that $k_i \leq k$ for all $i$, and by definition of Strahler number, there is at most one $m$ such that $k_m=k$), we define $\Gamma$ to be the set comprising the union of the $\Gamma_i$ and all the states of the form $\left(v, \seq{r_{k}, \ldots, r_1}\right)$ and $\left(v, \seq{r_{k}, \ldots , r_1} , p\right)$ with $v$ a vertex of $(A_i\setminus S_i) \cup A$ and $1\leq p \leq 2k$. 

  \textit{Case 1: } For each $i$, $k_i<k$. 
  
  We first show that $\Gamma$ is a trap for Audrey for the strategy defined above, showing that rank $2k+1$ can never be reached (implying that the strategy is defensive). This comes from the fact that the register of rank $k$ is only reset in a state $\left(v, \seq{r_k, \ldots, r_1}\right)$ with $v$ in $S$. Since $\max(r_k, d)= d$ is even then this leads to a state $\left(v, \seq{d,0, \ldots, 0}, 2k\right)$. Otherwise, register $k$ is never reset, so a state with rank $2k+1$ cannot be reached.

  Consider now any cycle in $\Reg{k}{\Gc}$ with moves restricted to the strategy constructed above. If this cycle contains a state whose first component is a vertex of $S$, then as explained above, the highest rank in the cycle is $2k$. Otherwise, the cycle is necessarily in $\Reg{k}{\Gc\cap S_i}$ for some $i$. By induction, $\tau_i$ is winning and so the cycle is even.
  
  \textit{Case 2: } There is a unique $m$ such that $k_m = k$. 
  
  We first show that a state of rank $2k+1$ is never reached. Observe that register $k$ is reset in two places: (1) immediately after a state with first component a vertex of $S$ is visited, (2) if register $k$ is reset by $\tau_m$. In the first case, similarly as shown above, a state of rank $2k$ is reached. In the second case, register $k$ is either reset in a state $\left(v, \seq{d, r_{k-1}, \ldots, r_1}\right)$, and similarly as above, a state of rank $2k$ is reached, or in a state of $\Omega_i$. In this case, as $\tau_i$ is defensive on $\Omega_i$ by induction, a state of rank $2k+1$ cannot be reached, and the highest rank that can be reached is $2k$.
  
  Proving that every cycle is even is similar to the previous case.
\end{proof}

\section{Strahler-Optimal Attractor Decompositions}

In this section we prove that every parity game whose Lehtinen number
is~$k$ has an attractor decomposition of Strahler number at most~$k$. 
In other words, we establish the Lehtinen number upper bound on the
Strahler number, which together with
Lemma~\ref{lem:Strahler-bounds-Lehtinen} provides a positive answer to 
Question~\ref{question:Strahler-eq-register}. 

\begin{theorem}
\label{thm:Lehtinen-bounds-Strahler}
  The Strahler number of a parity game is no larger than its Lehtinen
  number. 
\end{theorem}

When talking about strategies in parity games in
Section~\ref{section:tuning}, we only considered 
positional strategies, for which it was sufficient to verify the
parity criterion on (simple) cycles.
Instead, 
we explicitly consider the parity criterion on infinite paths here,
which we find more convenient to establish properties of Audrey
strategies in the proof of
Theorem~\ref{thm:Lehtinen-bounds-Strahler}. 

First, we introduce the concepts of \emph{tight} and
\emph{offensively optimal} attractor decompositions.

\begin{definition} 
\label{def:Tight}
  A Steven $d$-attractor decomposition $\Hc$ of $\Gc$ is \emph{tight}
  if Audrey has a winning strategy from at least one state in
  $\Def{\Strah{\Hc}-1}{\Gc}$ in which the value of register
  $\Strah{\Hc}-1$ is~$d$. 
\end{definition}

By definition, the existence of a tight Steven $d$-attractor
decomposition on a parity game implies that the Lehtinen number of the
game is at least its Strahler number, from which 
Theorem~\ref{thm:Lehtinen-bounds-Strahler} follows. 
Offensive optimality of an attractor decomposition, the concept we
define next, may seem less natural and more technical than tightness,
but it facilitates our proof that every game has a tight attractor
decomposition.

\begin{definition} 
  \label{def:Opt}
  Let   
  $\Hc = \seq{A,(S_1, \Hc_1, A_1), \ldots,(S_\ell, \Hc_\ell, A_\ell)}$
  be a Steven $d$-attractor decomposition, let games $\Gc_i$ for 
  $i = 1, 2, \dots, \ell$ be as in the definition of an attractor
  decomposition, let $A_i'$ be the Audrey attractor of the set of
  vertices of priority $d-1$ in $\Gc_i$, and let 
  $\Gc_i' = \Gc_i\setminus A_i'$.  
  We say that $\Hc$ is \emph{offensively optimal} if for every 
  $i = 1, 2, \ldots, \ell$, we have:  
  \begin{itemize}
  \item 
    Audrey has a dominion strategy on $\Def{\Strah{\Hc_i}-1}{\Gc'_i}$; 
  \item 
    Audrey has a dominion strategy on
    $\Def{\Strah{\Hc_i}}{\Gc_i'\setminus S_i}$. 
  \end{itemize} 
\end{definition}

Proving that every offensively optimal Steven attractor
decomposition is tight (Lemma~\ref{lemma:AudreyStrategyTight}), and
that every Steven dominion in a parity game has an offensively optimal 
Steven attractor decomposition (Lemma~\ref{lemma:TightStrahler}), will 
complete the proof of Theorem~\ref{thm:Lehtinen-bounds-Strahler}. We first give two propositions that will be useful in the proofs.

\begin{proposition}
   \label{prop:offensivetodominion}
   For every parity game $\Gc$ and non negative integer $k$, if Audrey has a dominion strategy from every state of $\Def{k}{\Gc}$ then Audrey has a dominion strategy on $\Reg{k}{\Gc}$. 
 \end{proposition}
 
 \begin{proof}
   For every state $s$ of $\Def{k}{\Gc}$, Audrey has a winning strategy $\tau_s$ on
  $\Def{k}{\Gc}$ starting in~$s$. 
  We construct a dominion strategy for her on
  $\Reg{k}{\Gc}$: 
  after every visit to a state of rank $2k+1$, Audrey
  follows $\tau_s$, where $s$ is the first state that follows on the
  path and whose rank is smaller than $2k+1$. 
  This defines a dominion strategy on
  $\Reg{k}{\Gc}$.  
 \end{proof}

\begin{proposition}
   \label{prop:dominioni}
   If 
   $\Hc = \seq{A, (S_1, \Hc_1, A_1), \dots, 
     (S_{\ell}, \Hc_{\ell}, A_{\ell})}$ 
   is an offensively optimal Steven $d$-attractor decomposition, then for
   every $i = 1, 2, \dots, \ell$, we have that Audrey has a dominion
   strategy on $\Reg{\Strah{\Hc_i}-1}{\Gc_i}$ (and also a dominion
   strategy on $\Def{\Strah{\Hc_i}-1}{\Gc_i}$). 
 \end{proposition}
 
 \begin{proof}
 Let $i$ in $\{1, 2, \dots, \ell\}$.
Consider the following strategy in $\Def{\Strah{\Hc_i}-1}{\Gc_i}$:
\begin{itemize}
\item 
  On the set of states 
  whose vertex components are in~$A'_i$, 
  Audrey follows a strategy induced by the 
  reachability strategy in $A'_i$ to a vertex of priority~$d-1$
  (picking any move if $v$ is of priority~$d-1$);  
\item 
  In states whose vertex component is in~$\Gc'_i$, 
  Audrey plays a $(k-1)$-register dominion strategy
  on~$\Def{\Strah{\Hc_i}-1}{\Gc'_i}$.  
  Such a strategy exists by the definition of offensive optimality.
\end{itemize}
This strategy is indeed an Audrey dominion strategy
on~$\Def{\Strah{\Hc_i}-1}{\Gc_i}$, because any play either visits a state whose
first component is a vertex in $A_i'$ infinitely often, or it
eventually remains in $\Def{\Strah{\Hc_i}-1}{\Gc'_i}$. 
In the former case, the play visits a state
whose first component is a vertex of priority $d-1$ infinitely often.  
In the latter case, the strategy is a dominion strategy
  on~$\Def{\Strah{\Hc_i}-1}{\Gc'_i}$. 

Finally, we use Proposition~\ref{prop:offensivetodominion} to turn this Audrey dominion strategy on~$\Def{\Strah{\Hc_i}-1}{\Gc_i}$ into an Audrey dominion strategy on~$\Reg{\Strah{\Hc_i}-1}{\Gc_i}$. 
 \end{proof}

\begin{lemma}
  \label{lemma:AudreyStrategyTight}
  Every offensively optimal Steven attractor decomposition is tight.
\end{lemma}

\begin{proof}
  Let  
  $\Hc = \seq{A, (S_1, \Hc_1, A_1), \ldots, 
    (S_\ell, \Hc_\ell, A_\ell)}$ 
  be an offensively optimal $d$-attractor decomposition of a parity
  game and let $k=\Strah{\Hc}$. 
  We construct a strategy for Audrey in~$\Def{k-1}{\Gc}$ that is
  winning for her from at least one state in which the value of
  register $k-1$ is $d$.   
  We define $\Gc_i'$ and $A_i'$ as in Definition~\ref{def:Opt}.

\textit{Case 1:} $\Strah{\Hc_i}=k$ for some unique $i$ in
$\{1,\ldots, \ell\}$. In this case, we show that Audrey has a dominion
strategy on $\Def{k-1}{\Gc_i}$. 
Since $\Gc_i$ is a trap for Steven in~$\Gc$, this gives the desired
result. This directly follows from Proposition~\ref{prop:dominioni}.

\textit{Case 2:} There are $1 \leq i<j \leq \ell$ such that
$\Strah{\Hc_i} = \Strah{\Hc_j} = k-1$. 
We construct a strategy for Audrey in~$\Def{k-1}{\Gc}$ that is winning
for her from all states in~$\Gc_j$ whose register~$k-1$ has
value~$d$. 
Firstly, since $\Hc$ is offensively optimal, Audrey has a dominion 
strategy on $\Def{k-1}{\Gc_i'\setminus S_i}$, denoted by~$\tau_i$, and
a dominion strategy on~$\Reg{k-2}{\Gc'_i}$, denoted by~$\tau'_i$. 
Moreover, by Proposition~\ref{prop:dominioni}, we have that Audrey has a 
dominion strategy, denoted by $\tau_j$, on $\Reg{k-2}{\Gc_j}$
(note that $\Gc_j$ is a trap for Steven in $\Gc$). 
Consider the following strategy for Audrey in $\Def{k-1}{\Gc}$,
starting from a state whose vertex component is in~$\Gc_j$ and
register~$k-1$ has value~$d$:
\begin{itemize}
\item 
  As long as the value of register~$k-1$ is larger than~$d-1$, Audrey
  follows the strategy induced by~$\tau_j$, while ignoring the value
  of register~$k-1$, as long as this value is larger than~$d-1$.
\item 
  If the value in register $k-1$ is at most $d-1$:
  \begin{itemize}
  \item 
    In states whose vertex component is in $A'_i$, Audrey follows a
    strategy induced by the reachability strategy from~$A'_i$ to a
    vertex of priority~$d-1$
    (picking any move if the vertex has priority $d-1$); 
  \item 
    In states whose vertex component is in $\Gc'_i \setminus S_i$ and
    whose register~$k-2$ has value at most~$d-2$, Audrey
    follows~$\tau_i$; 
  \item 
    In states whose vertex component is in~$\Gc'_i$ and whose register
    $k-1$ has value~$d-1$, Audrey follows the strategy induced
    by~$\tau'_i$, while ignoring the value of regiser~$k-1$. 
  \end{itemize}
\end{itemize}
Audrey plays any move if none of the above applies.

  We argue that this strategy is winning for Audrey
  in~$\Def{k-1}{\Gc}$ from states whose vertex component is
  in~$\Gc_j$ and register~$k-1$ has value~$d$. 
  Consider an infinite path that starts in such a state. 
  As long as register $k-1$ has value~$d$, Audrey follows~$\tau_j$.  
  If Steven never resets register $k-1$ then Audrey wins. 
  Otherwise, once register $k-1$ has been reset, its value is at
  most~$d-1$.  
  Note that $\Gc_j$ is included in $A'_i \cup (\Gc'_i \setminus S_i)$.  
  If register $k-1$ has a value smaller than $d-1$, and the play never
  visits a state whose vertex component is in~$A'_i$, then
  Audrey has followed $\tau_i$ along the play 
  (she has never left $\Gc'_i \setminus S_i$ as the only way for Steven
  to go out $\Gc'_i \setminus S_i$ is to go to $A'_i$) and wins.  
  Otherwise, the play visits a state whose vertex component is
  in~$A'_i$, and so it visits a state whose vertex component has
  priority~$d-1$, leading to a state in which register $k-1$ has 
  value~$d-1$.  
  Finally, if a state whose vertex component is in $A'_i$ is visited 
  infinitely many times then Audrey wins. 
  Otherwise, Audrey eventually plays according to~$\tau'_i$. 
  If Steven never resets register $k-1$ then Audrey wins. 
  Otherwise, if Steven resets register~$k-1$, which at this point has
  value~$d-1$, a state of rank $2k-1$ is visited and Audrey wins.  
\end{proof}

\begin{lemma}
  \label{lemma:TightStrahler}
  Every Steven dominion in a parity game has an offensively optimal
  Steven attractor decomposition. 
\end{lemma}

\begin{proof}
  Consider a parity game $\Gc$ which is a Steven dominion. Let $k$ be
  the Lehtinen number of $\Gc$ and let $d$ be the largest even value
  such that $\pi^{-1}(\{d,d-1\})\neq \emptyset$. 
  We construct an offensively optimal Steven attractor decomposition
  by induction.  

  If $d=0$, it is enough to consider $\seq{A, \emptyset}$, where $A$ is
  the set of all vertices in~$\Gc$. 

  If $d>1$, let $A$ be the Steven attractor of the set of vertices of
  priority $d$ in $\Gc$. Let $\Gc_0 = \Gc \setminus A$. If $\Gc_0 =
  \emptyset$ then $\seq{A,\emptyset}$ is an offensively optimal Steven
  attractor decomposition for~$\Gc$. Otherwise, $\Gc_0$ is a non-empty
  trap for Steven in~$\Gc$ and therefore $\Gc_0$ has a Lehtinen number
  at most $k$. 
  Let $A'$ be the Audrey attractor of all the vertices of priority
  $d-1$ in the sub-game $\Gc_0$ and let $\Gc'_0 = \Gc_0\setminus A'$. 

  Given a positive integer $b$, let $L^{b}$ be the largest dominion in
  $\Gc'_0$ such that Steven has a dominion strategy on
  $\Def{b}{\Gc'_0}$. We define $m$ to be the smallest number such that
  $L^{m}\neq \emptyset$ and let $S_0 = L^{m}$.  
  We show that $m \leq k$. 
  To prove this, we construct an Audrey dominion strategy on 
  $\Def{b}{\Gc_0}$ for all $b$ such that $L^b = \emptyset$. 
  Since the Lehtinen number of $\Gc_0$ is at most $k$, this implies
  that $m \leq k$.   
  The Audrey dominion strategy on $\Def{b}{\Gc_0}$, assuming 
  $L^{b} = \emptyset$, is as follows: 
  \begin{itemize}
  \item 
    If the vertex component of a state is in~$A'$ then Audrey uses the
    strategy in~$A'$ induced by the reachability strategy to vertices
    of priority~$d-1$;

  \item 
    If the vertex component of a state is in~$\Gc'_0$ then Audrey uses
    her dominion strategy on~$\Def{b}{\Gc'_0}$, which exists because the
    Steven dominion $L^{b}$ in $\Def{b}{\Gc'_0}$ is empty. 
  \end{itemize}  
  Any play following the above strategy and visiting infinitely often
  a state of $\Def{b}{\Gc_0 \cap A'}$ is winning for Audrey. 
  A play following the above strategy and remaining eventually in 
  $\Def{b}{\Gc'_0}$ is also winning for Audrey.  

  Let $\Hc_0$ be the $(d-2)$-attractor decomposition of $S_0$ obtained
  by induction. 
  In particular, $\Hc_0$ is offensively optimal. 

  Let $A_0$ be the Steven attractor to $S_0$ in $\Gc_0$ and 
  let~$\Gc_1 = \Gc_0 \setminus A_0$. 
  Subgame $\Gc_1$ is a trap for Steven and therefore it is a Steven
  dominion.  
  Let 
  $\Hc' = \seq{\emptyset, (S_1, \Hc_1, A_1), \ldots, 
    (S_\ell, \Hc_\ell, A_\ell)}$ 
  be an offensively optimal Steven $d$-attractor decomposition of
  $\Gc_1$ obtained by induction. 

  We claim that 
  $\Hc = \seq{A, (S_0, \Hc_0, A_0), (S_1, \Hc_1, A_1), \ldots, 
    (S_\ell, \Hc_\ell, A_\ell)}$ 
  is an offensively optimal Steven $d$-attractor decomposition
  of~$\Gc$.   
  Since $\Hc'$ is offensively optimal, it is enough to show that: 
  \begin{itemize}
  \item 
    Audrey has a dominion strategy on 
    $\Def{\Strah{\Hc_0}-1}{\Gc'_0}$, 
  \item 
    Audrey has a dominion strategy on
    $\Def{\Strah{\Hc_0}}{\Gc'_0 \setminus S_0}$. 
\end{itemize}

  Since $\Hc_0$ is offensively optimal, Audrey has a winning  
  strategy from at least one state in $\Def{\Strah{\Hc_0}-1}{S_0}$, by
  Lemma~$\ref{lemma:AudreyStrategyTight}$, and hence 
  $m \geq \Strah{\Hc_0}$. 
  
  So, by choice of $m$, Steven does not have a defensive dominion strategy
  on $\Def{\Strah{\Hc_0}-1}{\Gc'_0}$ from any state. This means that Audrey has a dominion strategy on $\Def{\Strah{\Hc_0}-1}{\Gc'_0}$.
  
  Moreover, by construction of $S_0$, Audrey has a dominion strategy
  on $\Def{m}{\Gc'_0 \setminus S_0}$. 
  This implies that Audrey has a dominion strategy on
  $\Def{\Strah{\Hc_0}}{\Gc'_0 \setminus S_0}$. 
\end{proof}

\section{Strahler-Universal Trees}

Our attention now shifts to tackling
Question~\ref{question:Strahler-algorithmic}.  
The approach is to develop constructions of small ordered trees into
which trees of attractor decompositions or of progress measures can
be embedded. 
Such trees can be seen as natural search spaces for dominion
strategies, and existing meta-algorithms such as the universal
attractor decomposition algorithm~\cite{JM20} and progress measure
lifting algorithm~\cite{Jur00,JL17} can use them to guide their
search, performed in time proportional to the size of the trees in the 
worst case. 

An ordered tree is \emph{universal} for a class of trees if all trees
from the class can be embedded into it.  
The innovation offered in this work is to develop optimized
constructions of trees that are universal for classes of trees whose
complex structural parameter, such as the Strahler number, is
bounded. 
This is in contrast to less restrictive universal trees introduced by
Czerwi\'nski et al.~\cite{CDFJLP19} and implicitly constructed by
Jurdzi\'nski and Lazi\'c~\cite{JL17}, whose sizes therefore grow
faster with size parameters, leading to slower algorithms.

Firstly, we give an inductive construction of Strahler-universal trees
and an upper bound on their numbers of leaves. 
Then we introduce labelled ordered trees, provide a succinct
bit-string labelling of the Strahler-universal trees, and give an
alternative and more explicit characterization of the
succinctly-labelled Strahler-universal trees. 
Finally, we argue how the succinct bit-string labelling of
Strahler-universal trees facilitates efficient computation of the
so-called ``level-$p$ successors'' in them, which is the key
computational primitive that allows using ordered trees to solve 
parity games.  
The constructions and techniques we develop here are inspired by and 
significantly refine those introduced by Jurdzi\'nski and
Lazi\'c~\cite{JL17}.

\subparagraph*{Strahler-Universal Trees and Their Sizes}
\label{subsec:Strahler-universal}

Intuitively, an ordered tree \emph{can be embedded in} another if the
former can be obtained from the latter by pruning some subtrees.  
More formally, the trivial tree~$\seq{}$ can be embedded in every
ordered tree, and $\seq{T_1, T_2, \dots, T_k}$ can be embedded
in $\seq{T'_1, T'_2, \dots, T'_{\ell}}$ if there are indices 
$i_1, i_2, \dots, i_k$ such that 
$1 \leq i_1 < i_2 < \cdots < i_k \leq \ell$
and for every $j = 1, 2, \dots, k$, we have that $T_j$ can be
embedded in~$T'_{i_j}$. 

An ordered tree is \emph{$(n, h)$-universal}~\cite{CDFJLP19} 
if every $(n, h)$-small ordered tree can be embedded in it. 
We define an ordered tree to be 
\emph{$k$-Strahler $(n, h)$-universal} if every $(n, h)$-small ordered
tree whose Strahler number is at most~$k$ can be embedded in it, and  
we give a construction of small Strahler-universal trees.

\begin{definition}[Trees $U_{t, h}^k$ and~$V_{t, h}^k$]
\label{def:U-and-V}
  For all $t \geq 0$, we define trees 
  $U_{t, h}^k$ (for all $h$ and~$k$ such that $h \geq k \geq 1$)
  and 
  $V_{t, h}^k$ (for all $h$ and~$k$ such that $h \geq k \geq 2$)
  by mutual induction:
  \begin{enumerate}
  \item
    if $h = k = 1$ then $U_{t, h}^k = \seq{}$; 

  \item
    if $h>1$ and $k=1$ then 
    $U_{t, h}^k = \seq{U_{t, h-1}^k}$; 

  \item
    \label{item:U-and-V--n-eq-1}
    if $h \geq k \geq 2$ and $t=0$ then 
    $U_{t, h}^k = V_{t, h}^k = \seq{U_{t, h-1}^{k-1}}$; 

  \item
    if $h \geq k \geq 2$ and $t \geq 1$ then
    $V_{t, h}^k = 
    V_{t-1, h}^k \cdot \seq{U_{t, h-1}^{k-1}} \cdot V_{t-1, h}^k$; 

  \item
    \label{item:U-and-V--h-eq-k}
    if $h = k \geq 2$ and $n \geq 2$ then $U_{t, h}^k = V_{t, h}^k$;
    
  \item
    \label{item:U-and-V--h-g-k}
    if $h > k \geq 2$ and $n \geq 2$ then
    $U_{t, h}^k = V_{t, h}^k \cdot \seq{U_{t, h-1}^k} \cdot V_{t, h}^k$. 
  \end{enumerate}
\end{definition}

For $g \geq 0$, let $I_g$ be the trivial tree, that is the tree with
exactly one leaf, of height~$g$.  
For example, $I_1 = \seq{}$ and 
$I_3 = \seq{\seq{\seq{}}} = \seq{\seq{\circ}}$.  
It is routine to verify that if $h \geq k=1$ or $t=0$ then 
$U_{t, h}^k = I_h$, 
and if $h \geq k \geq 2$ and $t=0$ then $V_{t, h}^k = I_h$.

\begin{lemma}
  \label{lem:U-n-h-k-Strahler-universal}
  For all $n \geq 1$ and $h \geq k \geq 1$, the ordered tree 
  $U_{\floor{\lg n}, h}^k$ is $k$-Strahler $(n, h)$-universal. 
\end{lemma}

\begin{proof}
  We say that a tree has \emph{weak Strahler number} at most $k$ if 
  every subtree rooted in a child of the root has Strahler number at
  most~$k-1$. 
  A tree is then \emph{weakly $k$-Strahler $(n, h)$-universal}
  if every $(n, h)$-small ordered tree whose weak Strahler number is
  at most~$k$ can be embedded in it.
  We proceed by induction on the number of leaves in an ordered tree
  and its height, 
  using the following strengthened inductive hypothesis:
  \begin{itemize}
  \item
    for all $n \geq 1$ and $h \geq k \geq 1$, 
    ordered tree $U_{\floor{\lg n}, h}^k$ is $k$-Strahler 
    $(n, h)$-universal; 
  \item
    for all $n \geq 1$ and $h \geq k \geq 2$, 
    ordered tree $V_{\floor{\lg n}, h}^k$ is weakly $k$-Strahler 
    $(n, h)$-universal. 
  \end{itemize}

  Let $T$ be an $(n, h)$-small ordered tree of Strahler number at
  most~$k$. 
  If $n=1$, $h=1$, or $k=1$, 
  then $T$ is the trivial tree (with just one leaf) of height at
  most~$h$, and hence it can be embedded 
  in~$U_{\floor{\lg n}, h}^k = I_h$, the trivial tree of height~$h$.  
  Likewise, if $h \geq k \geq 2$ and $n=1$, 
  then $T$ is the trivial tree of height at most~$h$, and hence it can
  be embedded in~$V_{\floor{\lg n}, h}^k = I_h$, the trivial tree of
  height~$h$.     

  Otherwise, we have that $T = \seq{T_1, \dots, T_j}$ for 
  some~$j \geq 1$. 
  We consider two cases:
  either $\Strah{T_i} \leq k-1$ for all $i=1, \dots, j$, or there
  is~$q$ such that $\Strah{T_q} = k$. 
  Note that by Proposition~\ref{prop:Strahler-small}, the latter case
  can only occur if $h > k$.  

  If $\Strah{T_i} \leq k-1$ for all $i = 1, \dots, j$, 
  then we argue that $T$ can be embedded in~$V_{\floor{\lg n}, h}^k$,  
  and hence also in~$U_{\floor{\lg n}, h}^k$, because 
  $V_{\floor{\lg n}, h}^k$ can be embedded 
  in $U_{\floor{\lg n}, h}^k$ by definition
  (see items~\ref{item:U-and-V--n-eq-1}., \ref{item:U-and-V--h-eq-k}.,
  and~\ref{item:U-and-V--h-g-k}.\ of Definition~\ref{def:U-and-V}).  
  Let~$p$ (a pivot) be an integer such that 
  both trees
  $T' = \seq{T_1, \dots, T_{p-1}}$ and 
  $T'' = \seq{T_{p+1}, \dots, T_j}$ 
  are $(\floor{n/2}, h)$-small. 
  Then by the strengthened inductive hypothesis, 
  each of the two trees $T'$ and~$T''$ can be embedded in
  tree~$V_{\floor{\lg \floor{n/2}}, h}^k = V_{\floor{\lg n}-1, h}^k$
  and tree $T_p$ can be embedded in~$U_{\floor{\lg n}, h-1}^{k-1}$. 
  It then follows that tree 
  $T = T' \cdot \seq{T_p} \cdot T''$ can be embedded in 
  $V_{\floor{\lg n}, h}^k = V_{\floor{\lg n}-1, h}^k \cdot 
  \seq{U_{\floor{\lg n}, h-1}^{k-1}} \cdot
  V_{\floor{\lg n}-1, h}^k$. 

  If $\Strah{T_q} = k$ for some~$q$ (the pivot), then we argue that
  $T$ can be embedded in~$U_{\floor{\lg n}, h}^k$. 
  Note that each of the two trees
  $T' = \seq{T_1, \dots, T_{q-1}}$ 
  and $T'' = \seq{T_{q+1}, \dots, T_j}$ 
  is $(n, h)$-small and all trees $T_1, \dots, T_{q-1}$ and
  $T_{q+1}, \dots, T_j$ 
  have Strahler numbers at most~$k-1$. 
  By the previous paragraph, 
  it follows that each of the two trees~$T'$ and~$T''$ can be embedded  
  in~$V_{\floor{\lg n}, h}^k$.  
  Moreover, tree $T_q$ is $(n, h-1)$-small and hence, by the
  inductive hypothesis, it can be embedded 
  in~$U_{\floor{\lg n}, h-1}^k$. 
  It follows that tree $T = T' \cdot \seq{T_q} \cdot T''$ can be
  embedded in 
  $U_{\floor{\lg n}, h}^k = 
  V_{\floor{\lg n}, h}^k \cdot \seq{U_{\floor{\lg n}, h-1}^k} \cdot
  V_{\floor{\lg n}, h}^k$.  
\end{proof}

\begin{lemma}
\label{lemma:size-of-U}
  For all $t \geq 0$, we have:
  \begin{itemize}
  \item
    if $h \geq k = 1$ then $\leaves{U_{t, h}^k} = 1$;
  \item
    if $h \geq k \geq 2$ then 
    $\leaves{U_{t, h}^k} 
    \: \leq \:
    2^{t + k} 
    {{t + k - 2} \choose {k-2}}
    {{h-1} \choose {k-1}}$.
  \end{itemize}
\end{lemma}

\begin{proof}
  The proof is by structural induction, 
  where the inductive hypothesis contains both the statement that for
  all $t \geq 0$ and $h \geq k \geq 2$, we have:
  \begin{equation}
  \label{eq:size-of-U}
    \leaves{U_{t, h}^k} 
    \: \leq \:
    2^{t + k} 
    {{t + k - 2} \choose {k-2}}
    {{h-1} \choose {k-1}}\,,
  \end{equation}
  and that for all $t \geq 0$ and $h \geq k \geq 2$, we have the
  following analogous bound on the number of leaves of
  trees~$V_{t, h}^k$:  
  \begin{equation}
  \label{eq:size-of-V}
    \leaves{V_{t, h}^k} 
    \: \leq \:
    2^{t + k-1} 
    {{t + k - 2} \choose {k-2}}
    {{h-2} \choose {k-2}}\,. 
  \end{equation}
  The following cases correspond to the six items in
  Definition~\ref{def:U-and-V}. 
  \begin{enumerate}
  \item
    \label{item:h-eq-0}
    If $h = k = 1$ then $\leaves{U_{t, h}^k} = \leaves{\seq{}} = 1$. 

  \item
    \label{item:h-g-1--k-eq-1}
    If $h > 1$ and $k = 1$ then a straightforward induction on~$h$ can
    be used to show that $\leaves{U_{t, h}^k} = 1$. 

  \item
    \label{item:h-geq-k-geq-2--t-e-0}
    If $h \geq k \geq 2$ and $t = 0$ then, again, a straightforward
    induction on~$h$ yields that 
    $\leaves{V_{t, h}^k} = 1 < 
    2^{t+k-1} {{t + k - 2} \choose {k-2}} {{h-2} \choose {k-2}}$
    and
    $\leaves{U_{t, h}^k} = 1 < 
    2^{t+k} {{t + k - 2} \choose {k-2}} {{h-1} \choose {k-1}}$. 

  \item
    Suppose that $h \geq k \geq 2$ and $t \geq 1$.

    Firstly, for $h \geq k = 2$ and $t \geq 0$, we slightly strengthen
    the inductive hypothesis~(\ref{eq:size-of-V}) to:
    \begin{equation}
      \label{eq:size-of-V-n-h-1}
      \leaves{V_{t, h}^2} \: \leq \: 2^{t + 1} - 1\,, 
    \end{equation}
    which we prove by induction on~$t$.
    Indeed, for $t = 0$ it follows from
    item~\ref{item:h-geq-k-geq-2--t-e-0}.\ above, and for $t \geq 1$,
    we have:      
    \begin{multline*}
      \leaves{V_{t, h}^2}
      \: = \: 
      \leaves{U_{t, h-1}^1} + 2 \cdot \leaves{V_{t-1, h}^2}
      \\
      \: \leq \:
      1 + 2\left(2^{(t-1)+1}-1\right)
      \: = \:
      2^{t+1}-1
      \: < \:
      2^{t+1} {{t} \choose {0}} {{h-2} \choose {0}}\,,
    \end{multline*}
    where the first inequality follows from
    items~\ref{item:h-eq-0}.\ or~\ref{item:h-g-1--k-eq-1}.\ above, and 
    from the strengthened inductive
    hypothesis~(\ref{eq:size-of-V-n-h-1}).  

    Secondly, for $h \geq k \geq 3$ and $t \geq 1$ we have: 
    \begin{multline*}
      \leaves{V_{t, h}^k}
      \: = \: 
      \leaves{U_{t, h-1}^{k-1}} + 2 \cdot \leaves{V_{t-1, h}^k}
      \\
      \: \leq \:
      2^{t+k-1} 
      {{t + k - 3} \choose {k-3}}
      {{h-2} \choose {k-2}}
      + 2 \cdot
      2^{t+k-2} 
      {{t + k - 3} \choose {k-2}}
      {{h-2} \choose {k-2}}
      \\
      \: = \:
      2^{t+k-1}
      \left[
        {{t + k - 3} \choose {k-3}}
        +
        {{t + k - 3} \choose {k-2}}
        \right]
           {{h-2} \choose {k-2}}
           \: = \:
           2^{t+k-1}
           {{t + k - 2} \choose {k-2}}
           {{h-2} \choose {k-2}}\,,
    \end{multline*}
    where the first inequality follows from the inductive hypothesis and
    the last equality follows from Pascal's identity.

  \item
    Suppose that $h = k \geq 2$ and $t \geq 1$. 
    Then we have:
    \begin{multline*}
    \leaves{U_{t, h}^k}
    \: = \:
    \leaves{V_{t, h}^k}
    \: \leq \:
    2^{t+k-1} {{t+k-2} \choose {k-2}} {{h-2} \choose {k-2}}
    \\
    \: < \:
    2^{t+k} {{t+k-2} \choose {k-2}} {{h-1} \choose {k-1}}\,,
    \end{multline*}
    where the first inequality follows by the inductive hypothesis
    and the other one from~$h = k$. 

  \item
    Suppose $h > k \geq 2$ and $t \geq 1$. 
    Then we have:
    \begin{multline*}
      \leaves{U_{t, h}^k}
      \: = \: 
      \leaves{U_{t, h-1}^k} + 2 \cdot \leaves{V_{t, h}^k}
      \\
      \: \leq \:
      2^{t+k} 
      {{t + k - 2} \choose {k-2}}
      {{h-2} \choose {k-1}}
      + 2 \cdot
      2^{t+k-1} 
      {{t + k - 2} \choose {k-2}}
      {{h-2} \choose {k-2}}
      \\
      \: = \:
      2^{t+k}
      {{t + k - 2} \choose {k-2}}
      \left[
        {{h-2} \choose {k-1}}
        +
        {{h-2} \choose {k-2}}
        \right]
      \: = \:
      2^{t+k}
      {{t + k - 2} \choose {k-2}}
      {{h-1} \choose {k-1}}\,,
    \end{multline*}
    where the first inequality follows from the inductive hypothesis and
    the last equality follows from Pascal's identity.
    \qedhere
  \end{enumerate}
\end{proof}

\begin{theorem}
  \label{thm:size-of-U-n-h-k}
  For $k \leq \lg n$, the number of leaves of the
  $k$-Strahler $(n, h)$-universal ordered trees 
  $U_{\floor{\lg n}, h}^k$ is  
  $n^{O(1)} \cdot \left({h}/{k}\right)^k
  = n^{{k \lg (h/k)}/{\lg n} + O(1)}$, 
  which is polynomial in~$n$ if
  $k \cdot \lg\left({h}/{k}\right) \: = \: O(\log n)$. 
  In more detail, the number is at most 
  $n^{c(n)} \cdot (h/k)^k$, where
  $c(n) = 5.45$ if $k \leq \lg n$, 
  $c(n) = 3 + o(1)$ if $k = o(\log n)$, 
  and $c(n) = 1 + o(1)$ if $k = O(1)$. 
\end{theorem}

\begin{remark}
  By Proposition~\ref{prop:Strahler-small} and
  Lemma~\ref{lem:U-n-h-k-Strahler-universal}, for all positive
  integers $n$ and~$h$, the tree $U_{\floor{\lg n}, h}^{\floor{\lg n}+1}$ is 
  $(n, h)$-universal. 
  Theorem~\ref{thm:size-of-U-n-h-k} implies that the number of leaves
  of $U_{\floor{\lg n}, h}^{\floor{\lg n}+1}$ is 
  $n^{\lg(h/{\lg n})+O(1)}$, 
  which matches the asymptotic number of leaves of $(n, h)$-universal
  trees 
  of Jurdzi\'nski and Lazi\'c~\cite[Lemma~6]{JL17}. 
  In particular, if $h = O(\log n)$ then 
  $\lg({h}/{\lg n}) = O(1)$, and hence the number of leaves of 
  $U_{\floor{\lg n}, h}^{\floor{\lg n}+1}$ is polynomial in~$n$. 
\end{remark}

\begin{proof}[Proof of Theorem~\ref{thm:size-of-U-n-h-k}]
  By Lemma~\ref{lem:U-n-h-k-Strahler-universal}, ordered 
  tree~$U_{\floor{\lg n}, h}^k$ is $k$-Strahler $(n, h)$-universal.  
  By Lemma~\ref{lemma:size-of-U}, its number of leaves is at most 
  $2^{\floor{\lg n}+k} {{\floor{\lg n}+k-2} \choose {k-2}} 
  {{h-1} \choose {k-1}}$. 

  We analyze in turn the three terms $2^{\floor{\lg n}+k}$, 
  ${\floor{\lg n}+k-2} \choose {k-2}$, and ${h-1} \choose {k-1}$. 
  Firstly, we note that 
  $2^{\floor{\lg n}+k}$ is $O\left(n^{p_1(n, k)}\right)$, where 
  $p_1(n, k) = 1+k/{\lg n}$, because $2^k = n^{k/{\lg n}}$. 
  Secondly, $k \leq \lg n$ implies that 
  $\floor{\lg n} + k - 2 < 2 \lg n$, 
  therefore we have 
  ${{\floor{\lg n}+k-2} \choose {k-2}} <
  2^{2 \lg n} \: = \: n^2$, 
  and hence 
  ${{\floor{\lg n}+k-2} \choose {k-2}}$ is $O(n^{p_2(n, k)})$,
  where $p_2(n, k) \leq 2$.
  Thirdly, applying the inequality 
  ${i \choose j} \leq \left(ei/j\right)^j$ 
  to the binomial coefficient 
  ${{h} \choose {k}}$, we obtain
  ${{h-1} \choose {k-1}} 
  \: < \:
  {{h} \choose {k}}
  \: \leq \:
  \left({eh}/{k}\right)^k
  \: = \:
  2^{k \lg(eh/k)}$,
  and hence 
  ${{h-1} \choose {k-1}}$ is $O(n^{p_3(n, h, k)})$, 
  where 
  $p_3(n, h, k) \: = \: {k \lg(eh/k)}/{\lg n}
  \: = \: {k \lg(h/k)}/{\lg n} + {k \lg e}/{\lg n}$. 

  Note that if we let 
  $p(n, h, k) = p_1(n, k) + p_2(n, k) + p_3(n, h, k)$ then the number of
  leaves in trees 
  $U_{\floor{\lg n}, h}^k$ is $O\!\left(n^{p(n, h, k)}\right)$.  
  Since $k \leq \lg n$ implies ${k}/{\lg n} \leq 1$
  and ${k \lg e}/{\lg n} \leq \lg e$, we obtain 
  $p(n, h, k) 
  \: \leq \: {k \lg(h/k)}/{\lg n} + 4 + \lg e
  \: < \: {k \lg(h/k)}/{\lg n} + 5.45$, 
  and hence the number of leaves in trees $U_{\floor{\lg n}, h}^k$ is  
  $n^{{k \lg(h/k)}/{\lg n} + O(1)}$. 

  If we further assume that $k = o(\log n)$ then the constant~$5.45$
  can be straightfowardly reduced to~$3+o(1)$ because then 
  ${k}/{\lg n}$ and ${k \lg e}/{\lg n}$ are~$o(1)$.  
  Moreover, the estimate 
  ${{\floor{\lg n}+k-2} \choose {k-2}} = O(n^2)$
  can be improved with further assumptions about~$k$ as a function
  of~$n$; 
  for example, if $k = O(1)$ then 
  ${{\floor{\lg n}+k-2} \choose {k-2}}$ 
  is only polylogarithmic in~$n$ and hence 
  ${{\floor{\lg n}+k-2} \choose {k-2}}$ is $n^{o(1)}$, 
  bringing $3+o(1)$ down to~$1+o(1)$. 
\end{proof}

\subparagraph*{Labelled Strahler-Universal Trees}

\emph{Labelled ordered tree} are similar to ordered trees: 
the trivial tree~$\seq{}$ is an \emph{$A$-labelled ordered tree} and 
so is a sequence  
$\seq{(a_1, \Lc_1), (a_2, \Lc_2), \dots, (a_k, \Lc_k)}$, where
$\Lc_1$, $\Lc_2$, \dots, $\Lc_k$ are $A$-labelled ordered trees, and 
$a_1$, $a_2$, \dots, $a_k$ are distinct elements of a linearly ordered  
set~$(A, \leq)$ and $a_1 < a_2 < \cdots < a_k$ in that linear
order.  
We define the \emph{unlabelling} of a labelled ordered tree 
$\seq{(a_1, \Lc_1), (a_2, \Lc_2), \dots, (a_k, \Lc_k)}$, 
by straightforward induction, to be the ordered tree
$\seq{T_1, T_2, \dots, T_k}$, where $T_i$ is the unlabelling
of~$\Lc_i$ for every $i = 1, 2, \dots, k$. 
An \emph{$A$-labelling} of an ordered tree~$T$ is an $A$-labelled 
tree~$\Lc$ whose unlabelling is~$T$. 
We define the \emph{natural labelling} of an ordered 
tree~$T = \seq{T_1, \dots, T_k}$, again by a straightfoward
induction, to be the $\Nats$-labelled tree 
$\seq{(1, \Lc_1), \dots, (k, \Lc_k)}$, where $\Lc_1$, \dots, $\Lc_k$
are the natural labellings of trees $T_1$, \dots, $T_k$.

For an $A$-labelled tree $\seq{(a_1, \Lc_1), \dots, (a_k, \Lc_k)}$,
its set of \emph{nodes} is defined inductively to consist of the
root~$\seq{}$ and all the sequences in~$A^*$ of the form $\seq{a_i}
\cdot v$, where $v \in A^*$ is a node in $\Lc_i$ for some 
$i = 1, \dots, k$, and where the symbol $\cdot$ denotes concatenation
of sequences.   
For example, the natural labelling of tree 
$\seq{\seq{\circ^3}, \circ^4, \seq{\seq{\circ}}^2}$ has the set of
nodes that consists of the following set of leaves 
$\seq{1, 1}$, $\seq{1, 2}$, $\seq{1, 3}$, $\seq{2}$, $\seq{3}$,
$\seq{4}$, $\seq{5}$,  $\seq{6, 1, 1}$, $\seq{7, 1, 1}$, and all of
their prefixes.  
Indeed, the set of nodes of a labelled ordered tree is always
prefix-closed.
Moreover, if $L \subseteq A^*$ then its closure under prefixes
uniquely identifies a labelled ordered tree that we call the
labelled ordered tree \emph{generated} by~$L$, and its unlabelling is
the ordered tree generated by~$L$. 
For example, the set 
$\eset{\seq{1}, \seq{3, 1}, \seq{3, 4, 1}, \seq{6, 1}}$
generates ordered tree 
$\seq{\circ, \seq{\circ, \seq{\circ}}, \seq{\circ}}$.

Consider the following linear order on the set $\eset{0, 1}^*$ of
bit strings:  
for each bit $b \in \eset{0, 1}$, and for all bit strings 
$\beta, \beta' \in \eset{0, 1}^*$, if $\varepsilon$ is the empty
string, then  we have $0\beta < \varepsilon$, $\varepsilon < 1\beta$,
and $b\beta < b\beta'$ iff $\beta < \beta'$. 

For a bit string $\beta \in \eset{0, 1}^*$, we write
$\left|\beta\right|$ for the number of bits used in the string.
For example, we have $\left|\varepsilon\right| = 0$ and
$\left|010\right| = 3$, and $\left|11\right| = 2$. 
Suppose that $\seq{\beta_i, \beta_{i-1}, \dots, \beta_1}$ is a
node in a $\eset{0, 1}^*$-labelled ordered tree.
Then if $\beta_j = b \beta$ for some $j = 1, 2, \dots, i$, 
$b \in \eset{0, 1}$, and $\beta \in \eset{0, 1}^*$, then we refer to
the first bit~$b$ as the \emph{leading bit} in~$\beta_j$, and we refer 
to all the following bits in~$\beta$ as \emph{non-leading bits}
in~$\beta_j$.  
For example, node 
$\seq{\varepsilon, 010, \varepsilon, \varepsilon, 11}$
has two non-empty strings and hence two leading bits, and it uses
three non-leading bits overall, because 
$\left|010\right| + \left|11\right| - 2 = 3$.

For a bit $b \in \eset{0, 1}$ and a 
$\eset{0, 1}^*$-labelled ordered tree 
$\Lc = \seq{\left(\beta_1, \Lc_1\right), \dots, 
  \left(\beta_\ell, \Lc_\ell\right)}$,
we define the 
$\eset{0, 1}^*$-labelled
ordered tree $\left[\Lc\right]^b$ to be equal to
$\Lc = \seq{\left(b \beta_1, \Lc_1\right), \dots, 
  \left(b \beta_\ell, \Lc_\ell\right)}$. 
In other words, $\left[\Lc\right]^b$ is the labelled ordered tree that
is obtained from~$\Lc$ by adding an extra copy of bit~$b$ as the
leading bit in the labels of all children of the root of~$\Lc$.

The inductive structure of the next definition is identical to that of
Definition~\ref{def:U-and-V}, and hence labelled ordered trees
$\Uc_{t, h}^k$ and~$\Vc_{t, h}^k$ defined here are labellings of the
ordered trees $U_{t, h}^k$ and~$V_{t, h}^k$, respectively. 

\begin{definition}[Trees $\Uc_{t, h}^k$ and $\Vc_{t, h}^k$]
\label{def:Uc-and-Vc}
  For all $t \geq 0$, we define $\eset{0, 1}^*$-labelled ordered
  trees~$\Uc_{t, h}^k$  
  (for all $h$ and~$k$ such that $h \geq k \geq 1$)
  and $\Vc_{t, h}^k$ (for all $h$ and~$k$ such that $h \geq k \geq 2$) 
  by mutual induction:
  \begin{enumerate}
  \item
    if $h = k = 1$ then 
    $\Uc_{t, h}^k = \seq{}$; 
  
  \item
    \label{item:h0-k0}
    if $h>1$ and $k=1$ then 
    $\Uc_{t, h}^k = 
    \seq{\tpl{\varepsilon, \Uc_{t, h-1}^k}}$; 
  
  \item
    \label{item:t0}
    if $h \geq k \geq 2$ and $t=0$ then 
    $\Vc_{t, h}^k = \seq{\tpl{\varepsilon, \Uc_{t, h-1}^{k-1}}}$  
    and $\Uc_{t, h}^k = \left[\Vc_{t, h}^k\right]^0 = 
    \seq{\tpl{0, \Uc_{t, h-1}^{k-1}}}$;
  
  \item
    if $h \geq k \geq 2$ and $t \geq 1$ then
    $\Vc_{t, h}^k = 
    \left[\Vc_{t-1, h}^k\right]^0 \cdot 
    \seq{\tpl{\varepsilon, \Uc_{t, h-1}^{k-1}}} \cdot 
    \left[\Vc_{t-1, h}^k\right]^1$; 
  
  \item
    \label{item:h-eq-k}
    if $h = k \geq 2$ and $t \geq 1$ then 
    $\Uc_{t, h}^k = \left[\Vc_{t, h}^k\right]^0$;
  
  \item
    \label{item:h-g-k}
    if $h > k \geq 2$ and $t \geq 1$ then
    $\Uc_{t, h}^k = 
    \left[\Vc_{t, h}^k\right]^0 \cdot 
    \seq{\left(\varepsilon, \Uc_{t, h-1}^k\right)} \cdot 
    \left[\Vc_{t, h}^k\right]^1$. 
  \end{enumerate}
\end{definition}

The inductive definition of labelled ordered trees $\Uc_{t, h}^k$
and~$\Vc_{t, h}^k$ makes it straightforward to argue that their
unlabellings are equal to trees $U_{t, h}^k$ and~$V_{t, h}^k$,
respectively, and hence to transfer to them Strahler-universality
established in Lemma~\ref{lem:U-n-h-k-Strahler-universal} and upper
bounds on the numbers of leaves established in
Lemma~\ref{lemma:size-of-U} and Theorem~\ref{thm:size-of-U-n-h-k}.  
We now give an alternative and more explicit characterization of those
trees, which will be more suitable for algorithmic purposes.  
To that end, we define $\eset{0, 1}^*$-labelled trees $\Bc_{t, h}^k$
and~$\Cc_{t, h}^k$ and then we argue that they are equal to
trees~$\Uc_{t, h}^k$ and~$\Vc_{t, h}^k$, respectively, by showing that
they satisfy all the recurrences in Definition~\ref{def:Uc-and-Vc}.

\begin{definition}[Trees $\Bc_{t, h}^k$ and $\Cc_{t, h}^k$]
  For all $t \geq 0$ and $h \geq k \geq 1$, we define 
  $\eset{0, 1}^*$-labelled ordered trees $\Bc_{t, h}^k$ as the tree
  generated by sequences $\seq{\beta_{h-1}, \dots, \beta_1}$ such
  that:  
  \begin{enumerate}
  \item
    the number of non-empty bit strings among $\beta_{h-1}$, \dots,
    $\beta_1$ is $k-1$;

  \item
    the number of bits used in bit strings $\beta_{h-1}$,
    \dots, $\beta_1$ overall is at most $(k-1)+t$;
  \end{enumerate}
  and for every $i = 1, \dots, h-1$, we have the following:
  \begin{enumerate}
  \setcounter{enumi}{2}
  \item
    if there are less than $k-1$ non-empty bit strings among
    $\beta_{h-1}$, \dots, $\beta_{i+1}$, but there are $t$ non-leading
    bits used in them, then $\beta_i = 0$;

  \item
    if all bit strings $\beta_i$, \dots, $\beta_1$ are non-empty, then  
    each of them has $0$ as its leading bit.
  \end{enumerate}

  For all $t \geq 0$ and $h \geq k \geq 2$, we define 
  $\eset{0, 1}^*$-labelled ordered trees $\Cc_{t, h}^k$ as the tree
  generated by sequences $\seq{\beta_{h-1}, \dots, \beta_1}$ such
  that:  
  \begin{enumerate}
  \item
    the number of non-empty bit strings among $\beta_{h-2}$, \dots,
    $\beta_1$ is $k-2$;

  \item
    the number of bits used in bit strings $\beta_{h-1}$,
    \dots, $\beta_1$ overall is at most $(k-2)+t$;
  \end{enumerate}
  and for every $i = 1, \dots, h-1$, we have the following:
  \begin{enumerate}
  \setcounter{enumi}{2}
  \item
    if there are less than $k-2$ non-empty bit strings among
    $\beta_{h-2}$, \dots, $\beta_{i+1}$, but there are
    $t-\left|\beta_{h-1}\right|$ non-leading bits used in them, then
    $\beta_i = 0$; 

  \item
    if all bit strings $\beta_i$, \dots, $\beta_1$ are non-empty, then  
    each of them has $0$ as its leading bit.
  \end{enumerate}
\end{definition}

\begin{lemma}
\label{lemma:Uc-eq-Bc}
  For all $t \geq 0$ and $h \geq k \geq 1$, we have 
  $\Uc_{t, h}^k = \Bc_{t, h}^k$. 
\end{lemma}

The following corollary follows from Lemma~\ref{lemma:Uc-eq-Bc}, and
from the identical inductive structures of
Definitions~\ref{def:U-and-V} and~\ref{def:Uc-and-Vc}. 

\begin{corollary}
\label{cor:Bc-eq-U}
  For all $t \geq 0$ and $h \geq k \geq 1$, 
  the unlabelling of 
  $\Bc_{t, h}^k$ is equal to 
  $U_{t, h}^k$. 
\end{corollary}

The next proposition formalizes the following non-rigorous
interpretation of the difference between trees $\Bc_{t, h}^k$
and~$\Cc_{t, h}^k$: 
\begin{itemize}
\item
  If a sequence $\seq{\beta_{h-1}, \dots, \beta_1}$ is a node
  in~$\Bc_{t, h}^k$ then the bit string $\beta_{h-1}$ can be either
  empty or non-empty, and if it is non-empty then its first bit is the
  leading bit.
\item
  On the other hand, if a sequence $\seq{\beta_{h-1}, \dots, \beta_1}$
  is a node in~$\Cc_{t, h}^k$ then the bit string $\beta_{h-1}$ is
  always to be understood as non-empty.
  It can be thought of as obtained by removal of its ``original''
  leading bit in the corresponding leaf in tree~$\Bc_{t, h}^k$, and
  hence it consists only of (possibly zero) non-leading bits.
\end{itemize}

\begin{proposition}
\label{prop:Bc-vs-Cc}
  For all $t \geq 1$ and $h \geq k \geq 2$, we have:
  \begin{enumerate}
  \item
    \label{item:Bc-vs-Cc--h-eq-k}
    if $h = k$ then
    $\seq{\beta_{h-1}, \dots, \beta_1}$ is a leaf in~$\Cc_{t, h}^k$ 
    if and only if
    $\seq{0 \beta_{h-1}, \beta_{h-2}, \dots, \beta_1}$ is a leaf
    in~$\Bc_{t, h}^k$; 

  \item
    \label{item:Bc-vs-Cc--h-g-k}
    if $h > k$ then for both $b \in \eset{0, 1}$, we have that 
    $\seq{\beta_{h-1}, \dots, \beta_1}$ is a leaf in~$\Cc_{t, h}^k$ 
    if and only if
    $\seq{b \beta_{h-1}, \beta_{h-2}, \dots, \beta_1}$ is a leaf
    in~$\Bc_{t, h}^k$;

  \item
    \label{item:Bc-vs-Cc--shorten}
    $\seq{\varepsilon, \beta_{h-2}, \dots, \beta_1}$ is a
    leaf in $\Cc_{t, h}^k$ 
    if and only if 
    $\seq{\beta_{h-2}, \dots, \beta_1}$ is a leaf 
    in~$\Bc_{t, h-1}^{k-1}$. 
  \end{enumerate}
\end{proposition}

\begin{proof}[{Proof of Lemma~\ref{lemma:Uc-eq-Bc}}]
  We argue that trees~$\Bc_{t, h}^k$ and~$\Cc_{t, h}^k$ satisfy all
  the recurrences in Definition~\ref{def:Uc-and-Vc} that involve trees
  $\Uc_{t, h}^k$ and~$\Vc_{t, h}^k$, respectively.

  \begin{enumerate}
  \item
    If $h = k = 1$ then
    tree $\Bc_{t, h}^k$ is the trivial tree~$\seq{}$. 

  \item
    If $h > k = 1$ then 
    $\Bc_{t, h}^k$ has only one leaf  
    $\seq{\varepsilon^{h-1}}$, 
    and hence we have 
    $\Bc_{t, h}^k = \seq{\tpl{\varepsilon, \Bc_{t, h-1}^k}}$. 

  \item
    Suppose that $h \geq k \geq 2$ and $t=0$.
    Then $\Bc_{t, h}^k$ has exactly one leaf, which is of
    the form $\seq{0^{k-1}, \varepsilon^{h-k}}$,
    and $\Cc_{t, h}^k$ has exactly one leaf, which is of the form 
    $\seq{\varepsilon, 0^{k-2}, \varepsilon^{h-k}}$. 
    It follows that 
    $\Cc_{t, h}^k = \seq{\tpl{\varepsilon, \Bc_{t, h-1}^{k-1}}}$
    and 
    $\Bc_{t, h}^k = \left[\Cc_{t, h}^k\right]^0 = 
    \seq{\tpl{0, \Bc_{t, h-1}^{k-1}}}$.  

  \item
    Suppose that $h \geq k \geq 2$ and $t \geq 1$.
    We argue that the following recurrence holds:
    \[
      \Cc_{t, h}^k 
      \: = \:
      \left[\Cc_{t-1, h}^k\right]^0 \cdot
      \seq{\tpl{\varepsilon, \Bc_{t, h-1}^{k-1}}} \cdot
      \left[\Cc_{t-1, h}^k\right]^1\,. 
    \]

    First, we show that every leaf in~$\Cc_{t, h}^k$ is
    also a leaf in tree 
    $\seq{\tpl{\varepsilon, \Bc_{t, h-1}^{k-1}}}$
    or in tree~$\left[\Cc_{t-1, h}^k\right]^b$ for 
    some~$b \in \eset{0, 1}$. 
    Suppose that 
    $\ell = \seq{\beta_{h-1}, \dots, \beta_1}$ 
    is  a leaf in~$\Cc_{t, h}^k$.
    \begin{itemize}
    \item
      If $\beta_{h-1} = \varepsilon$ then 
      $\seq{\beta_{h-2}, \dots, \beta_1}$ is a leaf 
      in~$\Bc_{t, h-1}^{k-1}$, and hence 
      $\ell = \seq{\varepsilon, \beta_{h-2}, \dots, \beta_1}$ 
      is a leaf in tree~$\seq{\tpl{\varepsilon, \Bc_{t, h-1}^{k-1}}}$. 

    \item
      If $\beta_{h-1} = b \beta$ for some $b \in \eset{0, 1}$ then
      $\seq{\beta, \beta_{h-2}, \dots, \beta_1}$ is a leaf
      in~$\Cc_{t-1, h}^k$, and hence 
      $\ell = \seq{b \beta, \beta_{h-2}, \dots, \beta_1}$ 
      is a leaf in~$\left[\Cc_{t-1, h}^k\right]^b$.  
    \end{itemize}

    Conversely, we now argue that if 
    $\ell = 
    \seq{\beta_{h-1}, \dots, \beta_1}$
    is a leaf in labelled ordered 
    tree $\seq{\tpl{\varepsilon, \Bc_{t, h-1}^{k-1}}}$, 
    then it is also a leaf in~$\Cc_{t, h}^k$.
    Note that the premise implies that $\beta_{h-1} = \varepsilon$ and 
    $\seq{\beta_{h-2}, \dots, \beta_1}$ 
    is a leaf in~$\Bc_{t, h-1}^{k-1}$, and hence, by
    item~\ref{item:Bc-vs-Cc--shorten}.\ in
    Proposition~\ref{prop:Bc-vs-Cc}, we have that 
    $\ell = \seq{\varepsilon, \beta_{h-2}, \dots, \beta_1}$ 
    is indeed a leaf in~$\Cc_{t, h}^k$.  

    Finally, we argue that if 
    $\ell = 
    \seq{\beta_{h-1}, \dots, \beta_1}$
    is a leaf in 
    a tree $\left[\Cc_{t-1, h}^k\right]^b$ for 
    $b \in \eset{0, 1}$, 
    then it is also a leaf in~$\Cc_{t, h}^k$.
    Indeed, the premise implies that $\beta_h = b \beta$ and 
    $\seq{\beta, \beta_{h-2}, \dots, \beta_1}$ 
    is a leaf in~$\Cc_{t-1, h}^k$, 
    and hence 
    $\ell = \seq{b \beta, \beta_{h-2}, \dots, \beta_1}$ is indeed a
    leaf in~$\Cc_{t, h}^k$.

  \item
    Suppose that $h = k \geq 2$ and $t \geq 1$.
    We argue that then we have 
    $\Bc_{t, h}^k = \left[\Cc_{t, h}^k\right]^0$. 

    First, let 
    $\ell = 
    \seq{\beta_{h-1}, \dots, \beta_1}$ be a leaf in 
    tree~$\Bc_{t, h}^k$.  
    Since $h = k$, all bit strings $\beta_{h-1}$, \dots, $\beta_1$ are
    non-empty, and hence $\beta_{h-1} = 0 \beta$ for some
    $\beta \in \eset{0, 1}^*$. 
    By item~\ref{item:Bc-vs-Cc--h-eq-k}.\ of
    Proposition~\ref{prop:Bc-vs-Cc}, it follows that the sequence 
    $\seq{\beta, \beta_{h-2}, \dots, \beta_1}$
    is a leaf in~$\Cc_{t, h}^k$, and hence 
    $\ell = \seq{0 \beta, \beta_{h-2}, \dots, \beta_1}$ 
    is indeed a leaf in~$\left[\Cc_{t, h}^k\right]^0$. 

    Conversely, let 
    $\ell = 
    \seq{\beta_{h-1}, \dots, \beta_1}$
    be a leaf in tree~$\left[\Cc_{t, h}^k\right]^0$. 
    Then $\beta_{h-1} = 0 \beta$ for some $\beta \in \eset{0, 1}^*$ 
    and sequence
    $\seq{\beta, \beta_{h-2}, \dots, \beta_1}$
    is a leaf in~$\Cc_{t, h}^k$.
    By item~\ref{item:Bc-vs-Cc--h-eq-k}.\ of
    Proposition~\ref{prop:Bc-vs-Cc}, it follows that 
    $\ell = \seq{0 \beta, \beta_{h-2}, \dots, \beta_1}$ is indeed a
    leaf in~$\Bc_{t, h}^k$.

  \item
    Suppose that $h > k \geq 2$ and $t \geq 1$.
    We argue that then the following recurrence holds:
    \[
      \Bc_{t, h}^k \: = \: 
      \left[\Cc_{t, h}^k\right]^0 \cdot 
      \seq{\left(\varepsilon, \Bc_{t, h-1}^k\right)} \cdot 
      \left[\Cc_{t, h}^k\right]^1\,.
    \]

    First, we show that every leaf in $\Bc_{t, h}^k$ is also a leaf in
    tree 
    $\seq{\left(\varepsilon, \Bc_{t, h-1}^k\right)}$
    or in tree $\left[\Cc_{t, h}^k\right]^b$ for some 
    $b \in \eset{0, 1}$. 
    Suppose that 
    $\ell = 
    \seq{\beta_{h-1}, \dots, \beta_1}$ is a leaf in~$\Bc_{t, h}^k$. 
    \begin{itemize}
    \item
      If $\beta_{h-1} = \varepsilon$ then 
      $\seq{\beta_{h-2}, \dots, \beta_1}$
      is a leaf in~$\Bc_{t, h-1}^k$, and hence 
      $\ell = \seq{\varepsilon, \beta_{h-2}, \dots, \beta_1}$ is a
      leaf in~$\seq{\left(\varepsilon, \Bc_{t, h-1}^k\right)}$.

    \item
      If $\beta_{h-1} = b \beta$ for some $b \in \eset{0, 1}$
      then, by item~\ref{item:Bc-vs-Cc--h-g-k}.\ of
      Proposition~\ref{prop:Bc-vs-Cc}, 
      $\seq{\beta, \beta_{h-2}, \dots, \beta_1}$ 
      is a leaf in~$\Cc_{t, h}^k$, and hence 
      $\ell = \seq{b \beta, \beta_{h-2}, \dots, \beta_1}$ 
      is a leaf in~$\left[\Cc_{t, h}^k\right]^b$. 
    \end{itemize}

    Conversely, we now argue that if 
    $\ell = 
    \seq{\beta_{h-1}, \dots, \beta_1}$
    is a leaf in labelled ordered 
    tree~$\seq{\left(\varepsilon, \Bc_{t, h-1}^k\right)}$,  
    then it is also a leaf in~$\Bc_{t, h}^k$. 
    Note that the premise implies that $\beta_{h-1} = \varepsilon$ and 
    $\seq{\beta_{h-2}, \dots, \beta_1}$
    is a leaf in~$\Bc_{t, h-1}^k$.
    It follows that 
    $\ell = \seq{\varepsilon, \beta_{h-2}, \dots, \beta_1}$  
    is indeed a leaf in~$\Bc_{t, h}^k$. 

    Finally, we argue that if 
    $\ell = 
    \seq{\beta_{h-1}, \dots, \beta_1}$
    is a leaf 
    in~$\left[\Cc_{t, h}^k\right]^b$ for some $b \in \eset{0, 1}$, 
    then it is also a leaf in~$\Bc_{t, h}^k$. 
    The premise implies that $\beta_{h-1} = b \beta$ for some
    $\beta \in \eset{0, 1}^*$ and that 
    $\seq{\beta, \dots, \beta_1}$ is a leaf in~$\Cc_{t, h}^k$.
    By item~\ref{item:Bc-vs-Cc--h-g-k}.\ of
    Proposition~\ref{prop:Bc-vs-Cc}, 
    it follows that 
    $\ell = \seq{b \beta, \beta_{h-2}, \dots, \beta_1}$ 
    is indeed a leaf in~$\Bc_{t, h}^k$. 
  \end{enumerate}

  Straightforward structural induction 
  (on the structure of labelled ordered trees~$\Uc_{t, h}^k$
  and~$\Vc_{t, h}^k$)
  yields that $\Bc_{t, h}^k = \Uc_{t, h}^k$ 
  and~$\Cc_{t, h}^k = \Vc_{t, h}^k$. 
\end{proof}

\subparagraph*{Efficiently Navigating Labelled Strahler-Universal Trees.}

The computation of the \emph{level-$p$ successor} of a leaf in a
labelled ordered tree of height~$h$ is the following problem:
  given a leaf 
  $\seq{\beta_h, \beta_{h-1}, \dots, \beta_1}$ in the tree 
  and given a number~$p$, 
  such that $1 \leq p \leq h$, compute the
  $<_{\mathrm{lex}}$-smallest leaf 
  $\seq{\beta'_h, \beta'_{h-1}, \dots, \beta'_1}$ in
  the tree, such that 
  $\seq{\beta_h, \dots, \beta_p}
  <_{\mathrm{lex}} \seq{\beta'_h, \dots, \beta'_p}$. 
As (implicitly) explained by Jurdzi\'nski and 
Lazi\'c~\cite[Proof of Theorem~7]{JL17}, the level-$p$ successor
computation is the key primitive used extensively in an implementation
of a progress measure lifting algorithm.

\begin{lemma}
\label{lemma:leaf-successor-poly-log}
  Every leaf in tree $\Bc_{t, h}^k$ can be
  represented using $O\left((k+t) \log h\right)$ bits and for every  
  $p = 1, 2, \dots, h$, the level-$p$ successor of a leaf in
  tree~$\Bc_{t, h}^k$ can be computed in time 
  $O\left((k+t) \log h\right)$.   
\end{lemma}

\begin{proof}% [Proof of Lemma~\ref{lemma:leaf-successor-poly-log}]
  Consider the following representation of a leaf 
  $\seq{\beta_{h-1}, \dots, \beta_1}$ in~$\Bc_{t, h}^k$: 
  for each of the at most $k+t$ bits used in the bit
  strings $\beta_{h-1}, \dots, \beta_1$ overall, store the value of
  the bit itself and the number, written in binary, of the component
  in the $h$-tuple that this bit belongs to.  
  Altogether, the number of bits needed is
  $O((k+t) \cdot (1 + \lg h)) \, = \, O((k+t) \log h)$. 
  
  We now consider computing the level-$p$ successor of a leaf 
  $\ell = \seq{\beta_{h-1}, \dots, \beta_1}$ in tree~$\Bc_{t, h}^k$. 
  We split the task of computing the level-$p$ successor 
  $\ell'$ of leaf $\ell$ into the following two steps:
  \begin{itemize}
  \item
    find the lowest ancestor $\seq{\beta_{h-1}, \dots, \beta_q}$ of 
    $\seq{\beta_{h-1}, \dots, \beta_p}$ 
    (that is, smallest $q$ satisfying $q \geq p$)
    that has the next sibling 
    $\seq{\beta_{h-1}, \dots, \beta_{q+1}, \beta'_q}$ in~$\Bc_{t, h}^k$; 
  \item
    find the smallest leaf 
    $\ell' = \seq{\beta_{h-1}, \dots, \beta_{q+1}, \beta'_q, \beta'_{q-1},
      \dots, \beta'_1}$
    that is a descendant of node 
    $\seq{\beta_{h-1}, \dots, \beta_{q+1}, \beta'_q}$ in~$\Bc_{t, h}^k$. 
  \end{itemize}

  For node $\ell_r = \seq{\beta_{h-1}, \dots, \beta_r}$, where 
  $q \leq r \leq h-1$, we can determine whether it has the next sibling 
  $\ell'_r = \seq{\beta_{h-1}, \dots, \beta_{r+1}, \beta'_r}$ 
  in~$\Bc_{t, h}^k$ and find it, by considering the following cases. 
  Firstly, we identify the cases in which $\ell_r$ does not have the  
  next sibling:
  \begin{itemize}
  \item
    the number of non-empty strings among $\beta_{h-1}$, \dots,
    $\beta_{r+1}$ is $k-1$;

  \item
    the number of non-leading bits used in strings $\beta_{h-1}$,
    \dots, $\beta_{r+1}$ is $t$;

  \item
    $\beta_r = 0 1^j$ for some $j \geq 0$, the number of non-leading
    bits used in strings $\beta_{h-1}$, \dots, $\beta_r$ is~$t$, and
    all bit strings $\beta_r$, \dots, $\beta_1$ are non-empty;

  \item
    $\beta_r = 1^j$ for some $j \geq 1$, and the number of non-leading
    bits used in strings $\beta_{h-1}$, \dots, $\beta_r$ is~$t$.
  \end{itemize}
  Define $k_{r+1}$ to be equal to $k-1$ minus the number of non-empty
  bit strings among $\beta_{h-1}$, \dots, $\beta_{r+1}$, and define 
  $t_{r+1}$ to be equal to $t$ minus the number of non-leading bits
  used in strings $\beta_{h-1}$, \dots, $\beta_{r+1}$. 
  We note that the subtree of~$\Bc_{t, h}^k$ that is rooted at node
  $\ell_{r+1}$ is a copy of tree~$\Bc_{t_{r+1}, r+1}^{k_{r+1}}$.
  Recall that trees $\Bc_{t, h}^k$ satisfy the same recurrences as
  trees~$\Uc_{t, h}^k$.
  Observe that the four cases above capture $\ell_r$ being the largest
  child of the root of the copy of~$\Bc_{t_{r+1}, r+1}^{k_{r+1}}$ rooted
  in node~$\ell_{r+1}$ in~$\Bc_{t, h}^k$, that correspond to
  items~\ref{item:h0-k0}., \ref{item:t0}., \ref{item:h-eq-k}.,
  and~\ref{item:h-g-k}.\ of Definition~\ref{def:Uc-and-Vc},
  respectively.  

  Secondly, we consider the remaining two cases in which $\ell_r$ does
  have the next sibling and we show how to find it by setting the
  value of $\beta_r'$ accordingly.
  \begin{itemize}
  \item
    If less than $t$ non-leading bits are used in strings
    $\beta_{h-1}$, \dots, $\beta_r$ then set 
    $\beta'_r = \beta_r 1 0^j$ for some $j \geq 0$, so that exactly
    $t$ non-leading bits are used in strings $\beta_{h-1}$, \dots,
    $\beta_{r+1}$, $\beta'_r$.

  \item
    If exactly $t$ non-leading bits are used in strings $\beta_{h-1}$,
    \dots, $\beta_r$, and $\beta_r = \beta 0 1^j$ for some 
    $\beta \in \eset{0, 1}^*$ and $j \geq 0$, then set 
    $\beta'_r = \beta$. 
  \end{itemize}

  Finally, we set 
  $\ell' \: = \: 
  \seq{\beta_{h-1}, \dots, \beta_{q+1}, \beta'_q, 00^i, 0, \dots,
    0, \varepsilon, \dots, \varepsilon}$
  for some suitable $i \geq 0$, so as to make the number of non-empty
  bit strings in~$\ell'$ equal to~$k-1$, and the number of bits used
  in all the bit strings in~$\ell'$ equal to $(k-1)+t$. 

  To argue that the above case analyses can be implemented to work in
  time $O((k+t)\log h)$, while using the succinct representation
  described above, is tedious and hence we eschew it.
\end{proof}

\section{Progress-Measure Strahler Numbers}

Consider a parity game~$\Gc$ in which all vertex priorities are at
most an even number~$d$.
If $(A, \leq)$ is a well-founded linear order then we write sequences
in $A^{d/2}$ 
in the following form
$\seq{m_{d-1}, m_{d-3}, \dots, m_1}$, 
and for every priority~$p \in \eset{0, 1, \dots, d}$, we define the 
\emph{$p$-truncation} of $\seq{m_{d-1}, m_{d-3}, \dots, m_1}$, denoted
by ${\seq{m_{d-1}, m_{d-3}, \dots, m_1}}|_p$, 
to be the
sequence $\seq{m_{d-1}, \dots, m_{p+2}, m_p}$ if $p$ is odd and 
$\seq{m_{d-1}, \dots, m_{p+3}, m_{p+1}}$ if $p$ is even. 
We use the lexicographic order~$\leq_{\mathrm{lex}}$ to linearly order
the set~$A^* = \bigcup_{i=0}^\infty A^i$. 

A \emph{Steven progress measure}~\cite{EJ91,Jur00,JL17} on a parity 
game~$\Gc$ is a map $\mu : V \to A^{d/2}$ such that for every
vertex~$v \in V$: 
\begin{itemize}
\item
  if $v \in V_{\Even}$ then there is a $\mu$-progressive edge 
  $(v, u) \in E$; 
\item
  if $v \in V_{\Odd}$ then every edge $(v, u) \in E$ is
  $\mu$-progressive; 
\end{itemize}
where we say that an edge $(v, u) \in E$ is \emph{$\mu$-progressive}
if: 
\begin{itemize}
\item
  if $\pi(v)$ is even then 
  ${\mu(v)}|_{\pi(v)} \geq_{\mathrm{lex}} {\mu(u)}|_{\pi(v)}$; 
\item
  if $\pi(v)$ is odd then 
  ${\mu(v)}|_{\pi(v)} >_{\mathrm{lex}} {\mu(u)}|_{\pi(v)}$. 
\end{itemize}
We define \emph{the tree of a progress measure~$\mu$} to be the
ordered tree generated by the image of~$V$ under~$\mu$. 

\begin{theorem}[\cite{EJ91,Jur00,JL17}]
\label{thm:pm-win-str}
  There is a Steven progress measure on a parity game~$\Gc$ if and
  only if every vertex in~$\Gc$ is in its largest Steven dominion.  
  If game $\Gc$ is $(n, d)$-small then the tree of a progress measure 
  on~$\Gc$ is $(n, {d}/{2} + 1)$-small.  
\end{theorem}

We define the \emph{Steven progress-measure Strahler number} of a
parity game~$\Gc$ to be the smallest Strahler number of a tree of a
progress measure on~$\Gc$. 
The following theorem refines and strengthens
Theorems~\ref{thm:attractor-decompositions-of-largest-dominia}
and~\ref{thm:pm-win-str} by establishing that the
Steven Strahler number and the Steven progress-measure Strahler number
of a parity game nearly coincide.

\begin{theorem}
\label{thm:ad-Strahler-eq-pm-Strahler}
  The Steven 
  Strahler number and the Steven progress-measure Strahler number of
  a parity game differ by at most~$1$. 
\end{theorem}
The translations between progress measures and attractor
decompositions are as given by Daviaud, Jurdzi\'nski, and
Lazi\'c~\cite{DJL18};  
here we point out that they do not increase the Strahler number of the 
underlying trees by more than~$1$. 
This coincidence of the two complexity measures, one based on
attractor decompositions and the other based on progress measures, 
allows us in Section~\ref{sec:coda} to use a progress measure
lifting algorithm to solve games with bounded Strahler number.

\begin{proof}[Proof of Theorem~\ref{thm:ad-Strahler-eq-pm-Strahler}]
Let $\Gc$ be a $(n,d)$-small parity game. To prove
Theorem~\ref{thm:ad-Strahler-eq-pm-Strahler} we will prove the
following two lemmas.  
\begin{lemma}
\label{lemma:dir1}
If $\Gc$ is a parity game where all the vertices belong to Audrey and $\Gc$ has a Steven attractor decomposition of Strahler number $k$, then it has a Steven progress measure of Strahler number at most $k +1$.
\end{lemma}

\begin{proof}
Let $\Gc$ be a parity game where all the vertices belong to Audrey.
The proof is by induction on the height of the tree of a Steven attractor decomposition of $\Gc$.

\subparagraph*{Induction hypothesis:} Given a $d$-attractor decomposition $\Hc$ of $\Gc$ and its tree $\Tc_\Hc$ of height $h$, there is a progress measure tree $\Tc$ of height $h$ and an embedding $f$ from $\Tc_\Hc$ to $\Tc$ such that all the nodes of $\Tc$ which are not in the image of $f$ are leaves.

\subparagraph*{Base case:} If the height of $\Tc$ is at most $0$, then the $d$-attractor decomposition is $(A,\emptyset)$. Let $C$ be the set of vertices which do not have priority $d$. Consider the topological order: $u<v$ if there is a path from $v$ to $u$ in $A$. We consider the tree $\seq{\circ^{|C|}}$ and $\mu$ which maps the vertices of priority $d$ to its root and the vertices in $C$ to leaves, respecting the topological order, i.e. if $u<v$ then $u$ is mapped to a node on the right of the node $v$ is mapped to. This defines a progress measure of Strahler number at most $2$.

\subparagraph*{Induction Step.}

Consider a Steven-$d$-attractor decomposition: 
$$\Hc \: = \: \seq{A, (S_1, \Hc_1, A_1), \dots, (S_j, \Hc_j, A_j)}$$ 
and let $\Tc_{\Hc_i}$ be the tree of $\Hc_i$ and  $\Gc_i$ as defined in the definition of an attractor decomposition.

Inductively, for all $i$, there is a progress measure tree $\Tc_i$ (and an associated progress measure mapping $\mu_i$) of the same height as $\Tc_{\Hc_i}$ and an embedding $f_i$ from $\Tc_{\Hc_i}$ to $\Tc_i$ such that all the nodes of $\Tc_i$ which are not in the image of $f_i$ are leaves. 

Let us construct a progress measure tree for $\Gc$ as follows. Let $C_i = A_i\setminus S_i$ for each $i$ and $C$ be the set of nodes in $A$ that have priority at most $d-1$. Set:
$$\Tc \: =  \: \seq{ \circ^{|C|} , \Tc_1, \circ^{|C_1|},\ldots, \Tc_j, \circ^{|C_j|}}$$

Set $\mu$ to be a mapping from the set of vertices of $\Gc$ to the nodes of $\Tc$ which extends $\mu_i$ on vertices in $S_i$, maps the vertices of priority $d$ to the root of the tree, the vertices in $C$ to the first $|C|$ children of the root and the vertices in $C_i$ to the corresponding $|C_i|$ children of the root which respects the topological ordering in $\Gc$ as viewed as a graph, i.e, if for vertices $u$ and $v$ in $C$, resp. $C_i$, there is a path from $u$ to $v$ in $C$, resp. $C_i$, then $u$ is mapped to a node that appears on the right of the node $v$ is mapped to.

By construction and induction hypothesis, the tree $\Tc$ embeds $\Tc_\Hc$ and the only nodes that are not images of nodes in $\Tc_\Hc$ are leaves. Moreover, $\Tc$ is a progress measure tree with mapping $\mu$ by induction hypothesis, and the construction which is compatible with the Steven reachability strategy on $A$, and the $A_i$'s.

The lemma follows from the fact that the Strahler number of a tree increases by at most 1 when leaves are added to it.
\end{proof}

\begin{lemma}
\label{lemma:dir2}
If $\Gc$ has a Steven progress measure of Strahler number $k$, then it has a Steven attractor decomposition of Strahler number at most $k$.
\end{lemma}

\begin{proof}

We will prove the following by induction, which proves the lemma:

\subparagraph*{Induction Hypothesis on $n$:} Given an $(n,d)$-small parity game $\Gc$ where $d$ is even and a progress measure tree $\Tc$ on $\Gc$, there exist a Steven attractor decomposition whose tree embeds in $\Tc$.

\begin{remark}
\label{remark:embed}
Given a progress measure mapping $\mu$ on $\Gc$ and its corresponding progress measure tree $\Tc$, and given a trap $R$ for Audrey in $\Gc$, the restriction of $\mu$ to the vertices in $R$ is a progress measure with the tree induced by the nodes images of the vertices of $R$ by $\mu$. 
\end{remark}

\subparagraph*{Base Case.} For games with one vertex, any progress measure tree on $\Gc$ and any tree of a Steven attractor decomposition are $\seq{}$. Therefore the induction hypothesis is satisfied.

\subparagraph*{Induction step.} 
Let $\Gc$ be an $(n,d)$-small parity game where $d$ is the least even integer no smaller than any priority in $\Gc$ and let $\Tc$ be a progress measure tree on $\Gc$.

\medskip

\noindent \textit{Case 1: If the highest priority in $\Gc$ is even, i.e. equal to $d$.}
Let $A$ be the Steven attractor of the set of vertices of priority $d$. Let $\Gc' = \Gc\setminus A$. As $\Gc'$ is a trap for Audrey in $\Gc$, the tree $\Tc'$ induced by the nodes images of the vertices in $\Gc'$ in $\Tc$ is a progress measure tree of $\Gc'$. By induction hypotheses, there exist a Steven attractor decomposition $\Hc$ of $\Gc'$ whose tree $\Tc_\Hc$ embeds in $\Tc'$. By appending $A$ to $\Hc$, one gets a Steven attractor decomposition of $\Gc$ of same tree $\Tc_\Hc$, which then embeds in $\Tc$.

\medskip

\noindent \textit{Case 2: If the highest priority in $\Gc$ is odd, i.e. equal to $d-1$.}

No vertex is mapped to the root in the progress measure tree $\Tc$.
Let $\Tc_0, \Tc_1, \ldots, \Tc_j$ be the subtrees, children of the root of $\Tc$.
Let us note that vertices of priority $d-1$ cannot be mapped to nodes in $\Tc_0$ as they would not have progressive outgoing edges if that was the case. Let $S_0$ be the set of vertices mapped to nodes in $\Tc_0$ and let $A_0$ be the Steven attractor of $S_0$ in $\Gc$. We can assume that $S_0$ is non-empty (otherwise we remove $\Tc_0$ from $\Tc$ and start again).

Let $\Gc' = \Gc\setminus A_0$. As $\Gc'$ is a subgame, trap for Audrey, the tree $\Tc'$ with subtrees $\Tc_1, \ldots, \Tc_j$ is a progress measure tree on $\Gc'$. By induction, one gets a Steven attractor decomposition:
$$\Hc' \: = \: \seq{\emptyset, (S_1, \Hc_1, A_1), \dots, (S_j, \Hc_j, A_j)}$$
whose tree embeds in $\Tc'$.

Now, let us prove that $S_0$ is a trap for Audrey. Let $u$ be in $S_0$ and $v$ be one of its successor. For $(u,v)$ to be progressive, $v$ has to be mapped to a node in $\Tc_0$ and is then in $S_0$. Since there is always an outgoing progressive edge for Steven's vertices and all edges of Audrey's vertices are progressive, we can conclude that $S_0$ is a trap for Audrey, is a sub-game, and $\Tc_0$ is a progress measure tree on it. By induction, one gets a Steven attractor decomposition $\Hc_0$ of $S_0$, whose tree embeds in $\Tc_0$.

We have proved that:
$$\Hc \: = \: \seq{\emptyset,(S_0,\Hc_0,A_1), (S_1, \Hc_1, A_1), \dots, (S_j, \Hc_j, A_j)}$$
is a Steven attractor decomposition of $\Gc$ whose tree embeds in $\Tc$.
\end{proof}

Lemma~\ref{lemma:dir2} gives one direction of the theorem. For the reverse direction, consider $\Gc$ a parity game and $\Hc$ a Steven attractor decomposition of Strahler number $k$. This decomposition induces a winning strategy for Steven (with exactly one edge going out any vertex owned by Steven in $\Gc$). Consider the restriction of $\Gc$ to this Steven strategy. This is a game where all the vertices belong to Audrey, and which has $\Hc$ as a Steven attractor decomposition. We can apply Lemma~\ref{lemma:dir2} and obtain a Steven progress measure of Strahler number at most $k + 1$. The progress measure thus obtained is  also a progress measure of $\Gc$, which concludes the proof. 
\end{proof}

\section{Strahler-Universal Progress Measure Lifting Algorithm}
\label{sec:coda}

Jurdzi\'nski and Lazi\'c~\cite[Section~IV]{JL17} have implicitly
suggested that the progress-measure lifting algorithm~\cite{Jur00} can
be run on any ordered tree and they have established the correctness
of such an algorithm if their \emph{succinct multi-counters trees} were 
used.  
This has been further clarified by Czerwi\'nski et 
al.~\cite[Section~2.3]{CDFJLP19}, who have explicitly argued that any
$(n, d/2)$-universal ordered tree is sufficient to solve an 
$(n, d)$-small parity game in this way. 
We make explicit a more detailed observation that follows using the
same standard arguments 
(see, for example, Jurdzi\'nski and Lazi\'c~\cite[Theorem~5]{JL17}). 

\begin{proposition}
\label{prop:output-of-lifting-algo}
  Suppose the progress measure-lifting algorithm is run on a parity
  game~$\Gc$ and on an ordered tree~$T$. 
  Let $D$ be the largest Steven dominion in~$\Gc$ 
  on which there is a Steven progress measure whose tree can be
  embedded in~$T$.  
  Then the algorithm returns a Steven dominion strategy on~$D$. 
\end{proposition}
An elementary corollary of this observation is that if the
progress-measure lifting algorithm is run on the tree of a progress 
measure on some Steven dominion in a parity game, then the algorithm
produces a Steven dominion strategy on a superset of that dominion. 
Note that this is achieved in polynomial time because the tree of a
progress measure on an $(n, d)$-small parity game is 
$(n, d/2)$-small and the running time of the algorithm is dominated by
the size of the tree~\cite[Section~IV.B]{JL17}.

\begin{theorem}
\label{thm:Strahler-pm-run-time}
  There is an algorithm for solving $(n, d)$-small parity games of
  Strahler number~$k$ in quasi-linear space and time 
  $n^{O(1)} \cdot (d/2k)^k = n^{{k \lg({d}/{k})}/{\lg n} + O(1)}$,
  which is polynomial in~$n$ if $k \cdot \lg(d/k) = O(\log n)$. 
\end{theorem}

\begin{proof}
  By Proposition~\ref{prop:Strahler-small}, we may assume that 
  $k \leq 1 + \lg n$. 
  In order to solve an $(n, d)$-small parity game of Steven Strahler 
  number~$k$, run the progress-measure lifting algorithm for Steven on 
  tree~$\Bc_{\floor{\lg n}, {d}/{2} + 1}^{k+1}$, which is $(k+1)$-Strahler  
  $(n, {d}/{2} + 1)$-universal by
  Lemma~\ref{lem:U-n-h-k-Strahler-universal} and
  Corollary~\ref{cor:Bc-eq-U}.
  By Theorem~\ref{thm:ad-Strahler-eq-pm-Strahler} and by
  Proposition~\ref{prop:output-of-lifting-algo}, the algorithm will
  then return a Steven dominion strategy on the largest Steven
  dominion. 
  The running time and space upper bounds follow from
  Theorem~\ref{thm:size-of-U-n-h-k}, by the standard analysis of 
  progress-measure lifting as in~\cite[Theorem~7]{JL17}, and by
  Lemma~\ref{lemma:leaf-successor-poly-log}. 
\end{proof}

\begin{remark}
\label{remark:2-sqrt-lg}
  We highlight the $k \cdot \lg(d/k) = O(\log n)$ criterion from
  Theorem~\ref{thm:Strahler-pm-run-time}  
  as offering a novel trade-off between two natural structural
  complexity parameters of parity games 
  (number of of priorities~$d$ and the Strahler/Lehtinen number~$k$) 
  that enables solving them in time that is polynomial in the number
  of vertices~$n$.   
  It includes as special cases both the $d < \lg n$ criterion of
  Calude et al.~\cite[Theorem~2.8]{CJKLS17} and the $d = O(\log n)$
  criterion of Jurdzi\'nski and Lazi\'c~\cite[Theorem~7]{JL17}
  (set $k = \floor{\lg n} + 1$ and use
  Propositions~\ref{prop:tree-of-decomposition-is-small} 
  and~\ref{prop:Strahler-small} to justify it),
  and the $k = O(1)$ criterion of Lehtinen~\cite[Theorem~3.6]{Leh18}  
  (by Theorem~\ref{thm:Lehtinen-bounds-Strahler}). 

  We argue that the new 
  $k \cdot \lg(d/k) = O(\log n)$ 
  criterion
  (Theorem~\ref{thm:Strahler-pm-run-time}) 
  enabled by our results 
  (coincidence of the Strahler and the Lehtinen
  numbers: Theorem~\ref{thm:Lehtinen-bounds-Strahler})  
  and techniques 
  (small and efficiently navigable Strahler-universal
  trees: 
  Theorem~\ref{thm:size-of-U-n-h-k}, 
  Corollary~\ref{cor:Bc-eq-U}, 
  and Lemma~\ref{lemma:leaf-successor-poly-log})  
  considerably expands the asymptotic ranges of the natural structural
  complexity parameters in which parity games can be solved in
  polynomial time.  
  We illustrate it by considering the scenario in which the rates of
  growth of both $k$ and $\lg d$ as functions of~$n$ are
  $O\!\left(\sqrt{\log n}\right)$, i.e.,
  $d$ is $2^{O\left(\sqrt{\log n}\right)}$.  
  Note that the number of priorities~$d$ in this scenario is allowed
  to grow as fast as $2^{b \cdot \sqrt{\lg n}}$ for an arbitrary positive
  constant~$b$, which is 
  significantly larger than    
  what is allowed by the $d = O(\log n)$ criterion of Jurdzi\'nski and
  Lazi\'c~\cite[Theorem~7]{JL17}. 
  Indeed, its rate of growth is much larger than any
  poly-logarithmic function of~$n$, because for every positive
  constant~$c$, we have $(\lg n)^c = 2^{c \cdot {\lg {\lg n}}}$, and  
  $c \cdot {\lg {\lg n}}$ is exponentially smaller 
  than~$b \cdot \sqrt{\lg n}$.  
  At the same time, the $O\!\left(\sqrt{\log n}\right)$ rate of
  growth  
  allowed in this scenario for the Strahler number~$k$ 
  substantially exceeds $k = O(1)$ required by
  Lehtinen~\cite[Theorem~3.6]{Leh18}.    
\end{remark}

% \newpage

\bibliography{parity}

\end{document}